\documentclass[12pt]{article}
\usepackage{graphicx} 

\usepackage{amsmath}
\usepackage{amsthm}
\usepackage{amssymb}
\usepackage{hyperref}
\usepackage{cleveref}
\usepackage{xcolor}
\usepackage[parfill]{parskip}
\usepackage[margin=1in]{geometry}
\usepackage{natbib}
\usepackage{subfigure}
\usepackage{svg}
\usepackage{tikz}
\usepackage{tikzit}
\usepackage[inline]{enumitem}
\usepackage[most]{tcolorbox}
\usepackage{booktabs}
\usepackage{graphicx}
\usepackage{thmtools}
\usepackage{thm-restate}
\usepackage[noend]{algpseudocode}
\usepackage{authblk}

\crefname{supp}{Supplement}{Supplements}

\setlist{leftmargin=5.5mm}

\tikzstyle{white rectangle}=[fill=white, draw=black, shape=rectangle, tikzit fill=white, tikzit draw=black]
\tikzstyle{invisible rectangle}=[fill=none, draw=none, shape=rectangle]
\tikzstyle{blue text}=[text={rgb,255: red,136; green,164; blue,189}, draw=none, fill=none]
\tikzstyle{image}=[fill={rgb,255: red,234; green,247; blue,255}, draw=black, shape=rectangle, tikzit fill={rgb,255: red,234; green,247; blue,255}, tikzit draw=black]
\tikzstyle{image prime}=[fill={rgb,255: red,184; green,236; blue,255}, draw=black, shape=rectangle, tikzit fill={rgb,255: red,184; green,236; blue,255}, tikzit draw=black]
\tikzstyle{payload}=[fill={rgb,255: red,255; green,247; blue,234}, draw=black, shape=rectangle, tikzit fill={rgb,255: red,255; green,247; blue,234}, tikzit draw=black]
\tikzstyle{e}=[fill={rgb,255: red,230; green,230; blue,255}, draw=black, shape=rectangle, tikzit fill={rgb,255: red,230; green,230; blue,255}, tikzit draw=black]
\tikzstyle{e prime}=[fill={rgb,255: red,185; green,186; blue,255}, draw=black, shape=rectangle, tikzit fill={rgb,255: red,185; green,186; blue,255}, tikzit draw=black]
\tikzstyle{sigma}=[fill={rgb,255: red,255; green,217; blue,164}, draw=black, shape=rectangle, tikzit fill={rgb,255: red,255; green,217; blue,164}, tikzit draw=black]
\tikzstyle{tf split}=[draw=black, shape=rectangle, rectangle split, rectangle split parts=2, rectangle split part fill={{green!20,red!20}}]
\tikzstyle{white gray}=[fill=white, draw=gray, shape=rectangle, minimum width=2cm]
\tikzstyle{gray text}=[fill=none, draw=none, text={rgb,255: red,162; green,162; blue,162}]

\tikzstyle{right arrow}=[->]
\tikzstyle{container edge}=[-, draw={rgb,255: red,136; green,164; blue,189}]
\tikzstyle{double arrow}=[<->]
\tikzstyle{dashed}=[-, dash dot, draw=black, thick]
\tikzstyle{dashed gray}=[-, draw={rgb,255: red,162; green,162; blue,162}, thick, dash dot]

\usepackage{mathtools}
\DeclarePairedDelimiter\abs{\lvert}{\rvert}%
\DeclarePairedDelimiter\norm{\lVert}{\rVert}%

\makeatletter
\let\oldabs\abs
\def\abs{\@ifstar{\oldabs}{\oldabs*}}
\let\oldnorm\norm
\def\norm{\@ifstar{\oldnorm}{\oldnorm*}}
\makeatother

\theoremstyle{definition}
\newtheorem{definition}{Definition}[section]
\newtheorem{theorem}{Theorem}[section]

\theoremstyle{remark}

\newtheorem{claim}{Claim}

\newcommand{\CSS}{\mathsf{SIG}}
\newcommand{\Gen}{\mathsf{Generate}}
\newcommand{\Sign}{\mathsf{Sign}}
\newcommand{\Verify}{\mathsf{Verify}}

\newcommand{\REF}{\mathsf{REF}}
\newcommand{\Embed}{\mathsf{Embed}}
\newcommand{\Compare}{\mathsf{Compare}}

\newcommand{\RPWS}{\mathsf{RPWS}}
\newcommand{\Watermark}{\mathsf{Watermark}}
\newcommand{\Detect}{\mathsf{Detect}}

\newcommand{\PGWS}{\mathsf{PGWS}}
\newcommand{\Encode}{\mathsf{Encode}}
\newcommand{\Decode}{\mathsf{Decode}}

\newcommand{\true}{\mathtt{true}}
\newcommand{\false}{\mathtt{false}}
\newcommand{\negl}{\mathsf{negl}}
\newcommand{\bits}{\{0,1\}}
\newcommand{\poly}{\mathsf{poly}}

\title{On the Difficulty of Constructing a Robust and Publicly-Detectable Watermark}

\author[1]{Jaiden Fairoze\thanks{Work completed while interning at Google DeepMind.}}
\author[2]{Guillermo Ortiz-Jiménez}
\author[2]{Mel Vecerik}
\author[3]{Somesh Jha}
\author[2]{Sven Gowal}
\affil[1]{University of California, Berkeley}
\affil[2]{Google DeepMind}
\affil[3]{University of Wisconsin–Madison}

\date{February 7, 2025}

\begin{document}

\maketitle

\begin{abstract}
    This work investigates the theoretical boundaries of creating publicly-detectable schemes to enable the provenance of watermarked imagery.
    Metadata-based approaches like C2PA provide unforgeability and public-detectability.
    ML techniques offer robust retrieval and watermarking.
    However, no existing scheme combines robustness, unforgeability, and public-detectability.
    In this work, we formally define such a scheme and establish its existence.
    Although theoretically possible, we find that at present, it is intractable to build certain components of our scheme without a leap in deep learning capabilities.
    We analyze these limitations and propose research directions that need to be addressed before we can practically realize robust and publicly-verifiable provenance.
\end{abstract}

\section{Introduction}

What online content is trustworthy?
Central to such a question is determining whether a piece of content is authentic.
The challenge is more pressing than ever given the widespread availability of Generative AI (GenAI) technology.
Powerful models from StabilityAI~\citep{rombach2022high}, OpenAI~\citep{achiam2023gpt}, Google DeepMind~\citep{reid2024gemini}, Anthropic~\citep{anthropic2024claude}, Meta~\citep{dubey2024llama} and Midjourney~\citep{midjourney} (among others) are able to produce content that can be difficult to distinguish from human-crafted content, even for experts~\citep{ha2024organic}.
This has led to a range of new provenance issues pertaining to trustworthiness, intellectual property, and accountability.


\textbf{Promising approaches.}
The main pathways to enabling provenance are metadata-based provenance such as the C2PA standard~\citep{c2pa2023coalition}, watermarking such as Steg.ai~\citep{stegai}, Digimarc~\citep{digimarc} or SynthID~\citep{synthid}, retrieval such as Turnitin Similarity~\citep{turnitin}, and ML-based detection (e.g., for synthetic content) such as GPTZero~\citep{gptzero}.
We consider three key properties of schemes for tracking 
provenance: unforgeability, robustness, and public-detectability.

\textit{Unforgeability.}
A provenance scheme is unforgeable if no adversary can produce content traced to a source without knowledge of that source's secret authentication key.
This property is crucial to real-world provenance: content should only be traceable to Alice if Alice enabled traceability with her secret key.
Presently, metadata-based provenance (e.g., C2PA) is the only approach that supports unforgeability due to its use of cryptographic digital signatures~\citep{rivest1978method}.

\textit{Robustness to accidental stripping.} 
Traced content is considered robust if it can be traced even if the content has undergone natural transformations during its lifetime.
In light of widespread GenAI tools, the primary application of provenance tools is to enable online content traceability.
In this setting, robustness is key: content such as text, audio, or images cannot be expected to retain its original form after initial distribution.
Watermarking and retrieval mechanisms currently enable transformation-robust traceability.

\textit{Public-detectability.} 
Public-detectability separates the authentication and verification functionalities.
Entities that hold a secret key can authenticate content such that verification can be performed with a corresponding public key.
The public key can be used by \textit{anyone} to verify that content originated with the secret key holder.
In the special case where robustness is not required, cryptographic digital signatures provide this exact functionality.
Metadata-based provenance is publicly-detectable due to its use of cryptographic signatures, but it is not robust to accidental stripping.

\textbf{This work.} 
In this paper, we study the possibility of uniting the high-security and trustworthiness of cryptographic tools with the powerful robustness of deep learning-based provenance.
In particular, we ask:
\begin{center}
    \textit{Is it possible to design an image watermark that\\(a) preserves the robustness of deep watermarks and\\(b) meets a well-defined notion of unforgeability and public-detectability?}
\end{center}

We analyze the theoretical and practical feasibility of constructing an image watermark with the following properties:
\begin{itemize}
    \item \textit{Cryptographic unforgeability.}
    It should be computationally infeasible to generate adversarial content watermarked with a key that the adversary does not control.
    This should hold even if the adversary has full information minus the secret key.
    
    \item \textit{Robustness to accidental stripping.}
    The watermark should persist even if the image is naturally transformed. This property ignores the adversarial setting and only needs to hold (on average) over naturally-occurring transformations.
    
    \item \textit{Publicly-detectable.} The detection procedure should not contain any secret information---there should be no detriment to publicly releasing it and allowing anyone unlimited access.
    
    \item \textit{Quality-preserving.} Watermarked images should be of similar quality to the original image.
\end{itemize}

While we focus on images, our results apply to any high-entropy data that supports post-hoc watermarking and robust embeddings, e.g., audio and video.

\textbf{Main contributions.} Our core contributions follow:
\begin{enumerate}
    \item In~\Cref{sec:unforge_wm}, we present a watermarking scheme that is provably unforgeable and publicly-detectable, but with limited robustness.
    This scheme is similar to metadata-based provenance but without metadata---the watermarked image is the same size as the input image.
    \item In~\Cref{sec:robust_wm}, we define the requirements for a robust, unforgeable, and publicly-detectable watermark.
    We prove that it can be constructed using cryptographic signatures, post-hoc watermarks, and robust embeddings as building blocks.
    In particular, the resulting public watermark is robust to transformations that underlying post-hoc watermark and robust embedding support.
    \item In~\Cref{sec:barriers}, we study the barriers to deploying our robust scheme.
    We find that state-of-the-art image embedding models are vulnerable to adversarial attacks that can force embeddings to collide.
    Despite this, we observe a weak-but-significant correlation between resistance to adversarial attacks and model performance.
    This suggests that if future models are able to better capture human vision, they may enjoy intrinsic adversarial robustness, thereby enabling robust and publicly-detectable watermarking.
\end{enumerate}

\section{Related Work}
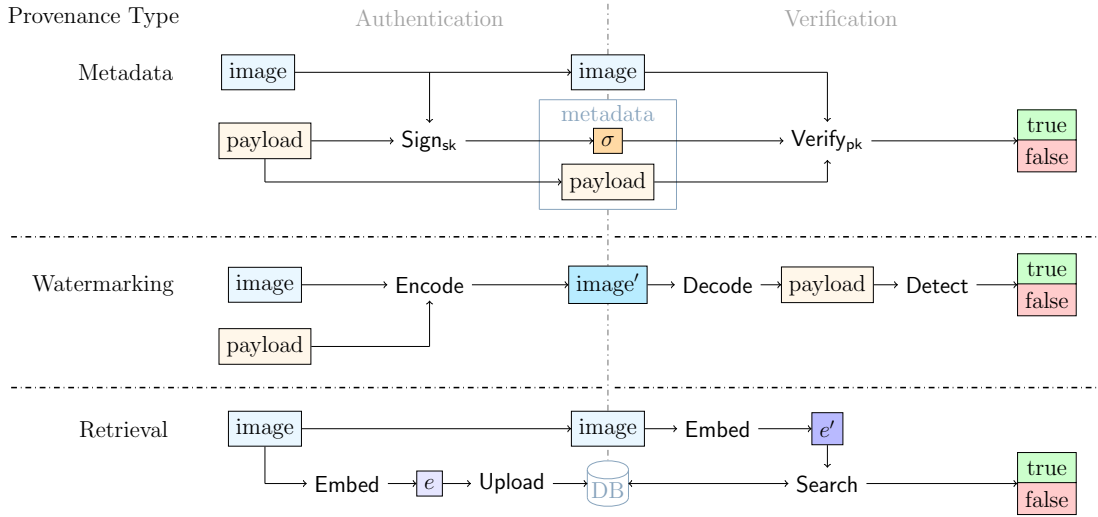
\begin{figure*}[t]
    \centering
    \resizebox{0.9\textwidth}{!}{%
        \begin{tikzpicture}
	\begin{pgfonlayer}{nodelayer}
		\node [style=none] (2) at (-6.575, -0.5) {Metadata};
		\node [style=image] (3) at (-1.75, -0.5) {$\mathrm{image}$};
		\node [style=image] (4) at (11, -0.5) {$\mathrm{image}$};
		\node [style=payload] (5) at (-1.5, -3) {payload};
		\node [style=sigma] (6) at (11, -3) {$\sigma$};
		\node [style=payload] (7) at (11, -4.5) {payload};
		\node [style=none] (8) at (-1.5, -4.5) {};
		\node [style=none] (9) at (8.5, -1.5) {};
		\node [style=none] (10) at (8.5, -5.5) {};
		\node [style=none] (11) at (13.5, -1.5) {};
		\node [style=none] (12) at (13.5, -5.5) {};
		\node [style=invisible rectangle] (14) at (4.5, -3) {$\mathsf{Sign_{sk}}$};
		\node [style=invisible rectangle] (15) at (19, -3) {$\mathsf{Verify_{pk}}$};
		\node [style=none] (16) at (19, -0.5) {};
		\node [style=none] (17) at (19, -4.5) {};
		\node [style=tf split] (18) at (27, -3) {true \nodepart{two} false};
		\node [style=none] (19) at (-7.4, -8.25) {Watermarking};
		\node [style=image] (20) at (-1.5, -8.25) {$\mathrm{image}$};
		\node [style=payload] (21) at (-1.5, -10.5) {payload};
		\node [style=invisible rectangle] (22) at (4.5, -8.25) {$\mathsf{Encode}$};
		\node [style=image prime] (23) at (11, -8.25) {$\mathrm{image}'$};
		\node [style=none] (24) at (4.5, -10.5) {};
		\node [style=none] (25) at (4.5, -0.5) {};
		\node [style=payload] (26) at (19, -8.25) {payload};
		\node [style=invisible rectangle] (27) at (15, -8.25) {$\mathsf{Decode}$};
		\node [style=tf split] (28) at (27, -8.25) {true \nodepart{two} false};
		\node [style=invisible rectangle] (29) at (23, -8.25) {$\mathsf{Detect}$};
		\node [style=none] (30) at (-6.65, -13.5) {Retrieval};
		\node [style=image] (31) at (-1.5, -13.5) {$\mathrm{image}$};
		\node [style=image] (32) at (11, -13.5) {$\mathrm{image}$};
		\node [style=e] (34) at (4.5, -15.5) {$e$};
		\node [style=none] (35) at (-1.5, -15.5) {};
		\node [style=none] (36) at (10.25, -15) {};
		\node [style=none] (37) at (10.25, -16) {};
		\node [style=none] (38) at (11.75, -16) {};
		\node [style=none] (39) at (11.75, -15) {};
		\node [style=blue text] (40) at (11, -15.75) {DB};
		\node [style=invisible rectangle] (41) at (1.5, -15.5) {$\mathsf{Embed}$};
		\node [style=invisible rectangle] (42) at (7.5, -15.5) {$\mathsf{Upload}$};
		\node [style=none] (43) at (10.25, -15.5) {};
		\node [style=invisible rectangle] (44) at (15, -13.5) {$\mathsf{Embed}$};
		\node [style=e prime] (45) at (19, -13.5) {$e'$};
		\node [style=tf split] (46) at (27, -15.5) {true \nodepart{two} false};
		\node [style=invisible rectangle] (47) at (19, -15.5) {$\mathsf{Search}$};
		\node [style=none] (48) at (11.75, -15.5) {};
		\node [style=blue text] (49) at (11, -2) {metadata};
		\node [style=none] (50) at (-10.75, -6.5) {};
		\node [style=none] (51) at (10.825, -6.5) {};
		\node [style=none] (52) at (-10.75, -12) {};
		\node [style=none] (53) at (10.825, -12) {};
		\node [style=none] (54) at (-7.75, 1.5) {Provenance Type};
		\node [style=gray text] (55) at (4.5, 1.5) {Authentication};
		\node [style=gray text] (56) at (19.5, 1.5) {Verification};
		\node [style=none] (57) at (11, 2) {};
		\node [style=none] (59) at (11, -1.5) {};
		\node [style=none] (60) at (11, -5.5) {};
		\node [style=none] (61) at (11, -14.65) {};
		\node [style=none] (62) at (29, -6.5) {};
		\node [style=none] (63) at (29, -12) {};
		\node [style=none] (64) at (11.25, -6.5) {};
		\node [style=none] (65) at (11.25, -12) {};
	\end{pgfonlayer}
	\begin{pgfonlayer}{edgelayer}
		\draw [style=dashed gray] (60.center) to (61.center);
		\draw [style=dashed gray] (57.center) to (59.center);
		\draw [style=container edge] (9.center) to (10.center);
		\draw (5) to (8.center);
		\draw [style=right arrow, in=180, out=0] (8.center) to (7);
		\draw [style=container edge] (9.center) to (11.center);
		\draw [style=container edge] (11.center) to (12.center);
		\draw [style=container edge] (10.center) to (12.center);
		\draw [style=right arrow] (5) to (14);
		\draw [style=right arrow] (14) to (6);
		\draw (7) to (17.center);
		\draw [in=180, out=0] (4) to (16.center);
		\draw [style=right arrow] (16.center) to (15);
		\draw [style=right arrow] (17.center) to (15);
		\draw [style=right arrow] (6) to (15);
		\draw [style=right arrow] (15) to (18);
		\draw [style=right arrow] (20) to (22);
		\draw [style=right arrow] (22) to (23);
		\draw (21) to (24.center);
		\draw [style=right arrow] (24.center) to (22);
		\draw [style=right arrow] (23) to (27);
		\draw [style=right arrow] (27) to (26);
		\draw [style=right arrow] (29) to (28);
		\draw [style=right arrow] (26) to (29);
		\draw [style=right arrow] (3) to (4);
		\draw [style=right arrow] (25.center) to (14);
		\draw [style=right arrow] (31) to (32);
		\draw [style=container edge, bend left=90, looseness=0.75] (36.center) to (39.center);
		\draw [style=container edge] (39.center) to (38.center);
		\draw [style=container edge] (36.center) to (37.center);
		\draw [style=container edge, bend right=90, looseness=0.75] (37.center) to (38.center);
		\draw [style=container edge, bend right=90, looseness=0.75] (36.center) to (39.center);
		\draw (31) to (35.center);
		\draw [style=right arrow] (35.center) to (41);
		\draw [style=right arrow] (41) to (34);
		\draw [style=right arrow] (42) to (43.center);
		\draw [style=right arrow] (34) to (42);
		\draw [style=right arrow] (32) to (44);
		\draw [style=right arrow] (44) to (45);
		\draw [style=double arrow] (48.center) to (47);
		\draw [style=right arrow] (47) to (46);
		\draw [style=right arrow] (45) to (47);
		\draw [style=dashed] (50.center) to (51.center);
		\draw [style=dashed] (52.center) to (53.center);
		\draw [style=dashed] (64.center) to (62.center);
		\draw [style=dashed] (65.center) to (63.center);
	\end{pgfonlayer}
\end{tikzpicture}
    }
    \caption{The three main approaches to content provenance. 
    Metadata-based provenance (top) uses an auxiliary manifest to attach a cryptographic signature and other metadata to the image---signature authentication yields provenance.
    Watermarking (middle) encodes a payload with provenance information directly into the image itself, and the payload can be decoded thereafter.
    Retrieval-based detection (bottom) maintains a global store of image embeddings where the store is queried to check if a candidate image is known.}
    \label{fig:misinformation_approaches}
\end{figure*}
We cover essential related work below.
For wider coverage, see Supplement~\ref{app:extended_rw}.
For a graphical overview of the main approaches to content provenance, see~\Cref{fig:misinformation_approaches}.

\textbf{Metadata-based provenance.}
Throughout the paper, we use metadata-based provenance to refer to strategies that attach signatures as additional metadata (See~\Cref{fig:misinformation_approaches}).
The C2PA~\citep{c2pa2023coalition} standard is the current leading approach.
The attached metadata is referred to as a manifest.
This manifest contains cryptographic information attesting to the content's creation, modification, and distribution history.
The manifest is ``hard bound'' to the content using a cryptographic hash---even a one bit change in content would cause its hash to change and detach the manifest from its content.
Both the content and manifest must be preserved in order for verification to succeed.

Metadata-based provenance can, in general, provide a comprehensive record of the content's history.
Due to its use of standardized cryptographic primitives, the system achieves a strong and well-defined notion of security.
Moreover, verifying a C2PA manifest is \textit{fully public}---anyone can verify the authenticity of a C2PA manifest by leveraging existing public-key infrastructure~\citep{laurie2014certificate}.

Conversely, the link between manifest and content is weak: it is trivial to detach the corresponding manifest from any arbitrary content.
Adversarial detachment aside, existing web infrastructure cannot readily support the manifest.
Updating infrastructure to accommodate C2PA manifests is a time-consuming and costly endeavor.
To partially address this problem, the C2PA working group is considering a \textit{soft binding} extension where metadata can be re-attached to content using a perceptual hash computed from the content or a watermark embedded within the digital content.
The group has referenced a candidate algorithm~\citep{iscc2024enhancement}, but any method that enables a ``similarity comparison'' between content is plausible.

\textbf{Watermarking.}
Fundamentally, watermarking schemes aim to hide information within content itself without visibly perturbing the content.\footnote{We do not consider ``visible'' watermarking schemes as they alter content and are easily strippable.}
The encoded information can enable provenance: in practice, the payload is usually a unique identifier for external information retrieval or, if the watermark capacity is large, a data store for origin-related information.

Common among watermarking schemes (see Supplement~\ref{app:extended_rw}) is their optimization to resist content modifications---the watermark payload should not be destroyed if transformed content is ``reasonably similar'' to the original watermarked image.
In addition, the watermark meets a high degree of imperceptibility: to an untrained eye, an image and its watermarked counterpart are of the same quality.
Since the watermark embeds information into the image itself rather than additional metadata, it does not require any web infrastructure changes and can be dropped in to enable provenance immediately.
However, such systems are subject to the following concerns:

 
\textit{Security.} 
There is no guarantee that proprietary algorithms are secure or correct, so end users cannot contextualize detection results.
In the text setting, post-hoc detectors commonly flag human-generated text as AI-generated, directly harming individuals~\citep{gegg2024ai}.

\textit{Utility.} All known industry watermarks only permit trusted users to plant or detect watermarks---the watermark provider cannot release the detector or perform watermarking client-side as it weakens security.

\textbf{Retrieval-based detection.} 
The most straightforward approach to provenance is to maintain a large, continuously-updated database containing every AI-generated image.
This solution is problematic when (a) different model providers do not have unified storage or (b) scalability issues arise once the database reaches a critical size.
Other issues (such as privacy) can be partially ameliorated by using \textit{fingerprints}: instead of storing images directly, a succinct and robust representation of each image (a fingerprint) is stored that preserves the ability to measure similarity.

Similarity comparisons arise in many areas of computer science beyond content fingerprinting~\citep{seo2004robust}, such as perceptual hashing~\citep{indyk1998approximate}, copy detection~\citep{chen2020simple}, and fuzzy matching~\citep{chaudhuri2003robust}.
In this work, we group these techniques under the umbrella of a ``robust embedding.''
In practice, robust embeddings are deployed for explicit material detection.
Examples are detection of non-consensual intimate image abuse~\citep{stopncii}, online terrorism~\citep{saltman2021practical}, and child sexual abuse material (CSAM)~\citep{apple2021csam,prokos2023squint}.

\section{Preliminaries}

We cover basic notation before presenting our schemes.
We also provide a succinct definition of cryptographic signatures as they are used throughout.

\textbf{Notation.}
Let $\lambda$ be the security parameter, i.e., the target level of security.
A system targeting $\lambda$ bits of security should be resistant to any attack that runs in at most $2^\lambda$ steps.
Let $\epsilon$ be a small error tolerance.
We will use $\epsilon$ to capture failure rates for various schemes.
Let $\poly(\cdot)$ refer to any arbitrary polynomial. 
Define $\negl(\lambda)$ to be a function such that for all $\poly(\lambda)$, it holds that $\negl(\lambda) < \frac{1}{\poly(\lambda)}$ for all sufficiently large $\lambda$.
We use superscripts to denote oracle access. For example, $A^O$ denotes that algorithm $A$ has oracle access to oracle $O$.

\textbf{Image transformations.}
Let $\mathcal{T}$ be the set of all possible transformation functions that can apply to an image.
We use $\Gamma(x)$ as the set of all transformations of $x$ and similarly $\Gamma(X)$ as the union of all sets of transformations of each element $\bigcup_{x \in X} \Gamma(x)$.

\subsection{Cryptographic Digital Signatures}
Given a secret key $sk$ and a public key $pk$ from a generation function $\Gen(1^\lambda)$\footnote{The security parameter $\lambda$ is passed as base-1 so that the time complexity of algorithm $\Gen$ is polynomial in the size of the input. For reference, see Chapter 3.1.1 in~\citet{katzlindell}.}, a signature $\sigma$ can be generated from content $x$ using the secret key: $\sigma \gets \Sign(sk, x)$.
This signature is verified by computing $\{\true,\false\} \gets \Verify(pk, x, \sigma)$.
The signature scheme must be \textit{correct} in the sense that honestly-generated signatures must verify with overwhelming probability. That is, for all $x$,
\begin{align*}
    \Pr\left[
    \begin{array}{ccc}
        \begin{matrix}
            \Verify(pk, x, \Sign(sk, x)) = \true: \\
            (sk, pk) \gets \Gen(1^\lambda) \\
        \end{matrix} \\
    \end{array}
    \right] \geq 1 - \negl(\lambda).
\end{align*}
Additionally, the scheme must satisfy a notion of \textit{unforgeability}\footnote{For most definitions, we use a weaker notion where the adversary does not have oracle access to relevant functions---this suffices for our purpose. We present the stronger versions in Supplement~\ref{app:formalism}.}, meaning for any probabilistic polynomial-time (PPT) adversary $\mathcal{A}$ and message $x$,
\begin{align*}
    \Pr\left[
    \begin{array}{ccc}
        \begin{matrix}
            \Verify(pk, x^*, \sigma^*) = \true \\
            \land\ x^* \neq x: \\
            (sk, pk) \gets \Gen(1^\lambda) \\
            \sigma \gets \Sign(sk, x) \\
            (x^*, \sigma^*) \gets \mathcal{A}(pk, x, \sigma) \\
        \end{matrix} \\
    \end{array}
    \right] \leq \negl(\lambda).
\end{align*}
That is, adversary $\mathcal{A}$ accepts its input (the public key $pk$, an honest message $x$, and a digital signature of the message $\sigma$). Its goal is to produce a pair $(x^*, \sigma^*)$ that ``break security'' such that (a) the pair is authentic (i.e., verification checks out) and (b) $x^*$ is a different message to the given $x$. If it succeeds, the adversary has forged a signature on a new message without the secret signing key.

\section{Warmup: An Unforgeable and Publicly-Detectable Watermark}\label{sec:unforge_wm}

Our first goal is to obtain a non-robust but unforgeable and publicly-detectable watermark. We ask:
\begin{center}
    \textit{Is it possible to obtain a scheme analogous to metadata-based provenance for images that does not introduce additional metadata?}
\end{center}

We present a simple scheme that embeds a cryptographic signature within an image $x$ such that there is a natural hash function satisfying $\mathsf{Hash}(x) = \mathsf{Hash}(x')$ where $x'$ is a visually-identical version of $x$ that embeds a cryptographic signature.
The scheme satisfies the equality by encoding signature bits into low-order bits of the image~\citep{muyco2019least}.
Image quality is guaranteed as pixel-values cannot change by more than 1, i.e., the PSNR is at worst $\approx 48.13$.
We note that the hash function is strongly collision-resistant for natural images: the hash of two visually-different images will not be the same.
We define our scheme in~\Cref{fig:unforgeable_scheme}.
\begin{figure*}[t]
\centering
\resizebox{0.9\textwidth}{!}{%
\begin{minipage}[t]{0.33\textwidth}
\begin{algorithmic}
\Function{$\mathsf{Watermark}$}{$sk, x$}
\State $\sigma \gets \Sign(sk, \mathsf{Hash}(x))$
\For{$x_{i, j, c}$ in $x$}
    \State $x'_{i, j, c} \gets 2 \cdot \lfloor\frac{x_{i, j, c}}{2}\rfloor + \sigma_{i, j, c}$
\EndFor
\State \textbf{return} $x'$
\EndFunction
\end{algorithmic}
\end{minipage}
\hfill\vline\hfill
\begin{minipage}[t]{0.33\textwidth}
\begin{algorithmic}
\Function{$\mathsf{Detect}$}{$pk, x$}
\State $h, \sigma \gets \mathsf{Hash}(x), \emptyset$
\For{$x_{i, j, c}$ in $x$}
    \State $\sigma \gets \sigma \parallel x_{i, j, c} \mod 2$
\EndFor
\State \textbf{return} $\Verify(pk, h, \sigma)$
\EndFunction
\end{algorithmic}
\end{minipage}
\hfill\vline\hfill
\begin{minipage}[t]{0.34\textwidth}
\begin{algorithmic}
\Function{$\mathsf{Hash}$}{$x$}
\State $h \gets \emptyset$
\For{$(r, g, b)$ in $x$}
    \State $h \gets h \parallel (\lfloor r / 2\rfloor, \lfloor g / 2\rfloor, \lfloor b / 2\rfloor)$
\EndFor
\State \textbf{return} $h$
\EndFunction
\end{algorithmic}
\end{minipage}
}
\caption{Specification of our unforgeable and publicly-detectable watermark. The keys are generated with the generation function of the signature scheme, $sk, pk \gets \Gen(1^\lambda)$. WLOG, the input image $x$ is RGB-encoded. The $\mathsf{Watermark}$ encodes a signature of the image within the image itself such that the output of $\mathsf{Hash}$ does not change. This is achieved by encoding signature bits in the least significant bit of each color channel value---when the hash is applied (i.e., each value is divided by two and floored), its value must be the same as the plain image.
Thus, $\mathsf{Detect}$ is able to recover both the hash value and signature bits in order to verify the signature.}
\label{fig:unforgeable_scheme}
\end{figure*}

\begin{theorem}[Informal]
The scheme presented in~\Cref{fig:unforgeable_scheme} is correct if the underlying cryptographic signature scheme is correct.
\end{theorem}

\begin{proof}
We claim that $\Detect(pk, \Watermark(sk, x))$ holds for all but negligibly few choices of $x$ and $sk, pk \gets \Gen(1^\lambda)$.
First, the underlying signature scheme is correct, meaning $\Verify(pk, h, \Sign(sk, h))$ holds for almost all inputs $h$.
Second, observe that the watermarking algorithm plants each bit of the signature into the lowest order integer bit of each color channel.
At detection time, this bit is extracted by computing $x_{i,j,c} \mod 2$ of each pixel channel.
\end{proof}

\begin{theorem}[Informal]
The scheme presented in~\Cref{fig:unforgeable_scheme} is unforgeable if the underlying cryptographic signature scheme is unforgeable.
\end{theorem}

\begin{proof}
To forge a watermark, it must be that either the signature scheme itself is forgeable, or the hash function is not collision-resistant.
It is given that the signature scheme is unforgeable, so it remains to see that the hash function is collision-resistant---for any two natural images, it should be negligibly likely that the hash of the images are the same.
This indeed holds: given any natural image $x \in \{0, \ldots, 255\}^n$, the $\ell_\infty$-norm ball with step $1$ \textit{must} be visually the same image.
In other words, changing all channel values by at most $1$ cannot visually change an image (we focus on images that show clear distinguishable semantic content).
\end{proof}


\section{A Robust and Publicly-Detectable Watermark}\label{sec:robust_wm}

Next, we augment our base scheme to add robustness.
In general, we rely on ML-based tools for robustness and cryptographic signatures for unforgeability and public-detectability.
We first define the ML-based tools, post-hoc watermarks and robust embeddings.

\subsection{Post-Hoc Watermarking} 
We treat post-hoc watermarks as a communication channel where the image is the channel and the watermark payload is the communicated data.
For this purpose, we only require a notion of \textit{correctness}: the decoded payload should be the same payload that was encoded---the worst attack that an adversary can launch is to destroy the payload.
We define post-hoc watermarks to satisfy the following interface: $\Gen(1^\lambda, \mathcal{T}) \to \Encode, \Decode$ is a possibly randomized algorithm that produces two functions, $\Encode(x, m) \to x'$ and $\Decode(x') \to m$, satisfying the following.
Let $x$ be an image and $m \in \bits^c$ be a $c$-length binary message (where $c$ represents the watermark capacity).
Then, for all choices of image $x$, message $m$, and transformation $T \in \mathcal{T}$,
\begin{align*}
    \Pr\left[
    \begin{array}{ccc}
        \begin{matrix}
            \Decode(T(\Encode(x, m))) = m: \\
            \Encode, \Decode \gets \Gen(1^\lambda, \mathcal{T}) \\
        \end{matrix} \\
    \end{array}
    \right] \geq 1 - \epsilon.
\end{align*}
That is, the payload is recovered with high probability.

\subsection{Robust Embedding Functions} 
This primitive captures various forms of similarity search (see Supplement~\ref{subsec:robust_embeddings} for details).
A robust embedding provides a possibly-randomized generation procedure $\Gen(1^\lambda, \mathcal{T}) \to \Embed, \Compare$ that yields two functions, $\Embed$ and $\Compare$, with respect to the set of transformations $\mathcal{T}$.
Given an image $x$, $e \gets \Embed(x)$ produces a succinct embedding $e$.
In order to compare the similarity of two images $x$ and $y$, it suffices to compare them in the embedding space by checking $\{\true,\false\} \gets \Compare(\Embed(x), \Embed(y))$ where the output of $\Compare$ is a binary value representing similarity.
Like other primitives, we require that embedding comparisons are correct: $\Compare$ should output $\true$ when the input embeddings were produced from visually similar images.
For all $x$ and transformations $T \in \mathcal{T}$,
\begin{align*}
    \Pr\left[
    \begin{array}{ccc}
        \begin{matrix}
            \Compare(e_1, e_2) = \true : \\
            \Embed, \Compare \gets \Gen(1^\lambda, \mathcal{T}) \\
            e_1 \gets \Embed(x) \\
            e_2 \gets \Embed(T(x)) \\
        \end{matrix} \\
    \end{array}
    \right] \geq 1 - \epsilon.
\end{align*}
Analogous to a cryptographic hash function, a robust embedding should satisfy a (weakened) notion of collision resistance where collisions are permitted if the input images are valid transformations as determined by the set of possible transformations $\mathcal{T}$.
We define collision resistance for robust embeddings as follows.
For any choice of $x$, it must be that
\begin{align*}
    \Pr\left[
    \begin{array}{ccc}
        \begin{matrix}
            \Compare(\Embed(x), \Embed(x^*)) = \true \\ 
            \land\ x^* \not\in \Gamma(x): \\
            \Embed \gets \Gen(1^\lambda, \mathcal{T}) \\
            e \gets \Embed(x) \\
            x^* \gets \mathcal{A}(\mathcal{T}, x, e) \\
        \end{matrix} \\
    \end{array}
    \right] \leq \epsilon.
\end{align*}

Note that a robust embedding is only useful if the embedding size is much smaller than the image size.
If not, it would be more efficient to compute similarity directly on the images.

\subsection{Construction}

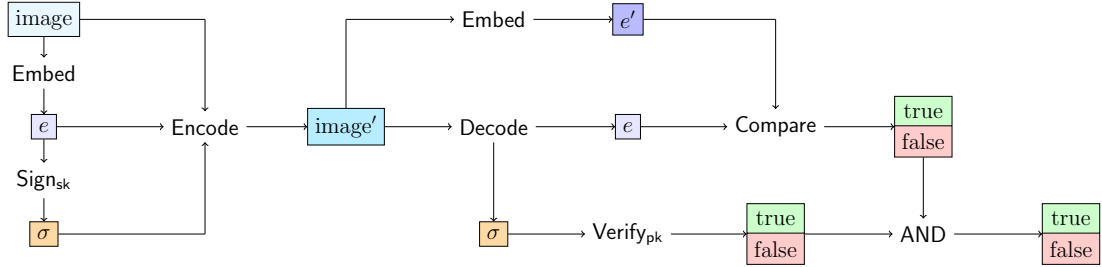
\begin{figure*}[t]
    \centering
    \resizebox{0.9\textwidth}{!}{%
        \begin{tikzpicture}
	\begin{pgfonlayer}{nodelayer}
		\node [style=image] (73) at (1.5, -0.5) {image};
		\node [style=e] (74) at (1.5, -4.5) {$e$};
		\node [style=sigma] (75) at (1.5, -8.5) {$\sigma$};
		\node [style=invisible rectangle] (76) at (1.5, -2.5) {$\mathsf{Embed}$};
		\node [style=invisible rectangle] (77) at (1.5, -6.5) {$\mathsf{Sign_{sk}}$};
		\node [style=invisible rectangle] (78) at (7.5, -4.5) {$\mathsf{Encode}$};
		\node [style=image prime] (79) at (12.75, -4.5) {$\mathrm{image}'$};
		\node [style=none] (80) at (7.5, -0.5) {};
		\node [style=none] (81) at (7.5, -8.5) {};
		\node [style=invisible rectangle] (82) at (18.25, -4.5) {$\mathsf{Decode}$};
		\node [style=invisible rectangle] (83) at (18.25, -0.5) {$\mathsf{Embed}$};
		\node [style=none] (84) at (12.75, -0.5) {};
		\node [style=sigma] (85) at (18.25, -8.5) {$\sigma$};
		\node [style=e] (86) at (23.25, -4.5) {$e$};
		\node [style=e prime] (87) at (23.25, -0.5) {$e'$};
		\node [style=invisible rectangle] (88) at (23.25, -8.5) {$\mathsf{Verify_{pk}}$};
		\node [style=tf split] (89) at (28.75, -8.5) {true \nodepart{two} false};
		\node [style=invisible rectangle] (90) at (28.75, -4.5) {$\mathsf{Compare}$};
		\node [style=none] (91) at (28.75, -0.5) {};
		\node [style=tf split] (92) at (34.25, -4.5) {true \nodepart{two} false};
		\node [style=invisible rectangle] (93) at (34.25, -8.5) {$\mathsf{AND}$};
		\node [style=tf split] (94) at (39.75, -8.5) {true \nodepart{two} false};
	\end{pgfonlayer}
	\begin{pgfonlayer}{edgelayer}
		\draw [style=right arrow] (73) to (76);
		\draw [style=right arrow] (76) to (74);
		\draw [style=right arrow] (74) to (77);
		\draw [style=right arrow] (77) to (75);
		\draw (73) to (80.center);
		\draw (75) to (81.center);
		\draw [style=right arrow] (81.center) to (78);
		\draw [style=right arrow] (80.center) to (78);
		\draw [style=right arrow] (74) to (78);
		\draw [style=right arrow] (78) to (79);
		\draw (79) to (84.center);
		\draw [style=right arrow] (84.center) to (83);
		\draw [style=right arrow] (79) to (82);
		\draw [style=right arrow] (82) to (85);
		\draw [style=right arrow] (82) to (86);
		\draw [style=right arrow] (83) to (87);
		\draw [style=right arrow] (85) to (88);
		\draw [style=right arrow] (88) to (89);
		\draw (87) to (91.center);
		\draw [style=right arrow] (86) to (90);
		\draw [style=right arrow] (91.center) to (90);
		\draw [style=right arrow] (90) to (92);
		\draw [style=right arrow] (89) to (93);
		\draw [style=right arrow] (92) to (93);
		\draw [style=right arrow] (93) to (94);
	\end{pgfonlayer}
\end{tikzpicture}
    }
    \caption{A robust and publicly-detectable watermark built from a cryptographic signature scheme, a post-hoc watermarking scheme, and a robust embedding model for images. Using a post-hoc watermark, the scheme encodes an embedding of the image along with a signature of the embedding within the image itself. This information can be decoded and verified thereafter.}
    \label{fig:this_work}
\end{figure*}

We define a watermarking scheme that is simultaneously publicly-detectable, unforgeable and robust.
See~\Cref{fig:this_work} for a graphical overview of our method.

\textbf{Approach.} 
From our base scheme, we incrementally build our final scheme with all desired properties.
First, we replace the pixel-level encoding procedure of our initial scheme with a post-hoc watermarking scheme: instead of hiding payload bits in the lower order bits of an image, we use a post-hoc watermark to plant the payload.
To watermark, we compute $x' \gets \Encode(x, \Sign(sk, x))$ and to verify the watermark we compute $b \gets \Verify(pk, x', \Decode(x'))$.
Unfortunately, this scheme does not work as-is: since $x' \neq x$, the signature is computed on a different image to the one produced by the post-hoc watermark\footnote{We discuss training a robust embedding such that $\Embed(x) = \Embed(x')$ in~\Cref{para:stable_hash}.}.

Instead of signing the image directly, we embed it $e \gets \Embed(x)$ and sign the embedding: $\sigma \gets \Sign(sk, e)$.
Then, we encode the signature and embedding into the image: $x' \gets \Encode(x, \sigma \parallel e)$.
To verify the watermark, we decode $\sigma$ and $e$ from the candidate watermarked image by parsing $\sigma, e \gets \Decode(x')$.
We also require the embedding of the candidate image: $e' \gets \Embed(x')$.
For the candidate image to be considered watermarked by $sk$, two conditions must be met.
First, $\Compare(e, e') = \true$: the watermarked image is perceptually similar to the image corresponding to the signed embedding.
Second, $\Verify(pk, e, \sigma) = \true$: the signature must be authentic.

\subsection{Properties}

\textbf{Threat model.}
The key difference between publicly- and privately-detectable watermarks arises at detection time: in the public setting, $\Detect$ takes a public key rather than the secret key used at watermarking time.
The adversary also has full details of the scheme except for the secret key: this includes white-box access to $\Sign$ and $\Verify$ from the signature scheme, $\Encode$ and $\Decode$ from the post-hoc watermark, and $\Embed$ and $\Compare$ from the robust embedding.
For security, an adversary should not be able to forge a watermark corresponding to a secret key out of her control. That is, for all $x$,
\begin{align*}
    \Pr\left[
    \begin{array}{ccc}
        \begin{matrix}
            \Detect(pk, x^*) = \true \\
            \land\ x^* \neq x': \\
            (sk, pk) \gets \Gen(1^\lambda) \\
            x' \gets \Watermark(sk, x) \\
            (x^*) \gets \mathcal{A}(pk, x') \\
        \end{matrix} \\
    \end{array}
    \right] \leq \epsilon.
\end{align*}

\begin{restatable}{thm}{main}\label{thm:ref_to_rpws}
If $(\mathcal{T}_\REF, m, n, \epsilon_\REF)$-robust embedding functions, $(\mathcal{T}_\PGWS, c, \epsilon_\PGWS)$-post-hoc watermarking schemes, and $(\delta, \lambda)$-cryptographic signatures exist such that $c \geq \delta + n$ and $\PGWS.\Encode \in \mathcal{T}_\REF$, then $(\mathcal{T}_\REF \cap \mathcal{T}_\PGWS, \epsilon_\REF + \epsilon_\PGWS + \negl(\lambda))$-publicly-detectable watermarking schemes also exist.
\end{restatable}
For the full proof of~\Cref{thm:ref_to_rpws}, see Supplement~\ref{app:embed_to_watermark}.
Given secure building blocks (i.e., robust embedding, post-hoc watermark, and signature scheme), the resulting watermark inherits key properties from each relevant underlying primitive.
We provide an overview below:

\textbf{Robustness.} The robustness of the resulting scheme is the set of transformations common between the robust embedding and the post-hoc watermark.
For the watermark to be detectable, the data encoded in the private watermark and the robust embedding of the image need to be preserved.

\textbf{Unforgeability.} We outline the high-level intuition for the unforgeability of our scheme.
Imagine the definition does not hold, and it is computationally tractable to find a forged $x^*$ for any input image $x$.
Given the forgery $x^*$ satisfying $\Verify(pk, e, \sigma) \land \Compare(e, e') = \true$, the forged message-signature pair $(e, \sigma)$ or colliding embeddings $(e, e')$ are immediately recoverable, implying attacks against the underlying primitives: either
\begin{enumerate*}[label=(\alph*)]
  \item the underlying cryptographic signature scheme is forgeable, or
  \item the underlying robust embedding is not collision-resistant
\end{enumerate*}.
Since we assume that such primitives are secure to begin with, we reach a contradiction implying that the watermark is indeed unforgeable.

\textbf{Imperceptibility.} Imperceptibility is an intrinsic property of the (underlying) watermarking scheme---we demonstrate how to use an existing private watermarking scheme to instantiate a publicly-detectable one. 
For example, if we use the TrustMark-Q~\citep{bui2023trustmark} watermark to instantiate our robust and public watermark, the scheme would inherit the PSNR ($43.26 \pm 1.59$) and SSIM ($0.99 \pm 0.00$) of TrustMark-Q directly. 
We further note that there is a tradeoff between watermark robustness or capacity and imperceptibility: watermarks with better robustness and/or capacity tend to introduce more distortions.

\section{Instantiating Our Scheme in the Real World}\label{sec:barriers}

We analyze the feasibility of instantiating each building block of our construction in the real world.

\subsection{(Compact) Cryptographic Signatures}

Standard cryptographic signature schemes such as RSA~\citep{rivest1978method} or ECDSA~\citep{johnson2001elliptic} are secure, efficient, and widely supported.
For use in a robust and publicly-detectable watermark, we require compact signature schemes with short signature lengths.
For standard 128-bit security, RSA signatures are at least 3072 bits ~\citep{elaine2016recommendation}, and ECDSA signatures are at least 512 bits (on the $\mathsf{secp256r1}$ elliptic curve; \citealp{certicom2010sec2}).
BLS signatures may offer better compactness~\citep{boneh2001short} though require stronger assumptions and careful choice of pairing-friendly curves.
For example, the most widely-used curve, $\mathsf{BLS12\text{-}381}$~\citep{barreto2003constructing}, supports a minimum signature size of 384 bits for 128-bit security.
In our setting, it is reasonable to target lower security as breaking even 80-bit security is expected to be orders of magnitude more difficult than launching attacks on the robust embedding.
This would require future work on suitable elliptic curves.

\begin{tcolorbox}[colback=gray!10!white,leftrule=2.5mm,size=title]
\textbf{Takeaway.} Barring a leap in compact digital signature design, signature sizes are unlikely to shorten drastically in the foreseeable future.
\end{tcolorbox}

\subsection{Realizing a Robust Embedding Function}

We evaluate a range of self-supervised image embedding models for their suitability to instantiate our robust embedding function.
We ask:

\begin{center}
    \textit{How collision-resistant are state-of-the-art image embedding models in a white-box setting?}
\end{center}

\textbf{Methodology.}
For a range of state-of-the-art embedding models (see ~\Cref{tab:models}), we instantiate a robust embedding function:
\begin{enumerate*}[label=(\alph*)]
    \item the embedding function is simply the model's native forward pass, and
    \item embeddings are compared by computing their $\ell_2$-normalized dot product.
\end{enumerate*}
Note that we do not binarize comparison values in order to capture fine-grained performance.
See Supplement~\ref{sec:experimental_data} for additional experimental data.

\textbf{Data.} We use the ``original'' and ``strong'' components of the Copydays dataset~\citep{douze2009evaluation} for copy detection.
This provides base images and their strongly-transformed variants.
For all base images, we randomly select a positive and negative image which respectively represent a transformed version of the original image and a completely different image.
We denote the similar pair as a \textit{positive} pair and the different pair as a \textit{negative} pair.
Thus, let $(I_{1,a}, I_{1,a'}), (I_{2,a}, I_{2,a'}), \ldots, (I_{1,a}, I_{1,b}), (I_{2,a}, I_{2,b}), \ldots \in \mathcal{D}$ be the dataset $\mathcal{D}$ of both positive $(I_{i, a}, I_{i, a'})$ and negative $(I_{i, a}, I_{i, b})$ image pairs.

\textbf{Attack.} We fix a baseline projected gradient descent (PGD)~\citep{madry2018towards} with momentum attack (20 steps) for various choices of $L_p$ epsilon values for $p \in \{1, \infty\}$. 
The attack is designed to force the $\ell_2$-normalized dot product of embedding pairs to be far (close) for positive (negative) pairs.
If an adversary has access to the image embedding, then using adversarial attacks, they can forge a watermark by forcing two different images to have a high similarity.
For example, if an adversary wants to use PGD and they have access to the embedding, they can perform gradient ascent over the input space up to some $\ell_p$ imperceptible norm bound.

\textbf{Evaluation.} We measure the following:
\begin{enumerate}
    \item \textit{Clean performance.} 
    We calculate the area under the receiver-operating characteristic curve (ROC AUC) of the binary classification problem captured by $\mathcal{D}$.
    We compute the score for a given pair with $\Compare(\Embed(I_{i,\cdot}), \Embed(I_{i,\cdot}))$ and assign a binary label depending on if the pair is positive or negative.
    This captures raw model performance as a robust embedding function.
    
    \item \textit{Attacked performance.} 
    The dataset $\mathcal{D}$ is attacked with fixed PGD attacks for various perturbation strengths such that positive (negative) pairs are embedded far (close) in the embedding space.
    ROC AUC is calculated as in the clean case but on the attacked dataset $\mathcal{A}_{\ell_p, \epsilon}(\mathcal{D})$ for $p \in \{1, \infty\}$ and $\epsilon \in \{1/255, 2/255, 4/255, 8/255, 16/255, 32/255\}$.
\end{enumerate}

\begin{figure*}[t]
\centering
\subfigure{
  \includegraphics[height=0.31\textwidth]{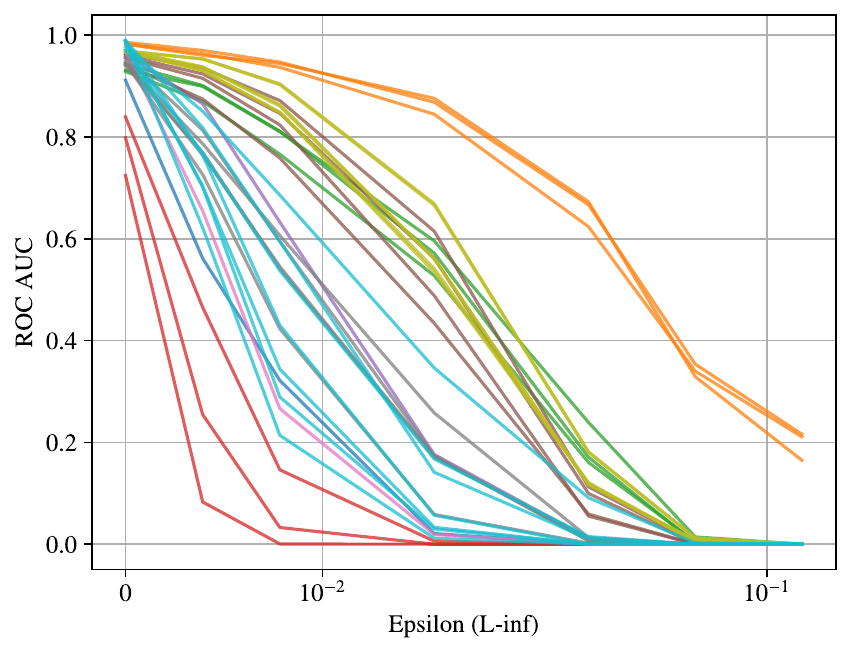}
  \label{fig:roc_auc_vs_epsilon_linf}
}
\subfigure{
  \includegraphics[height=0.31\textwidth]{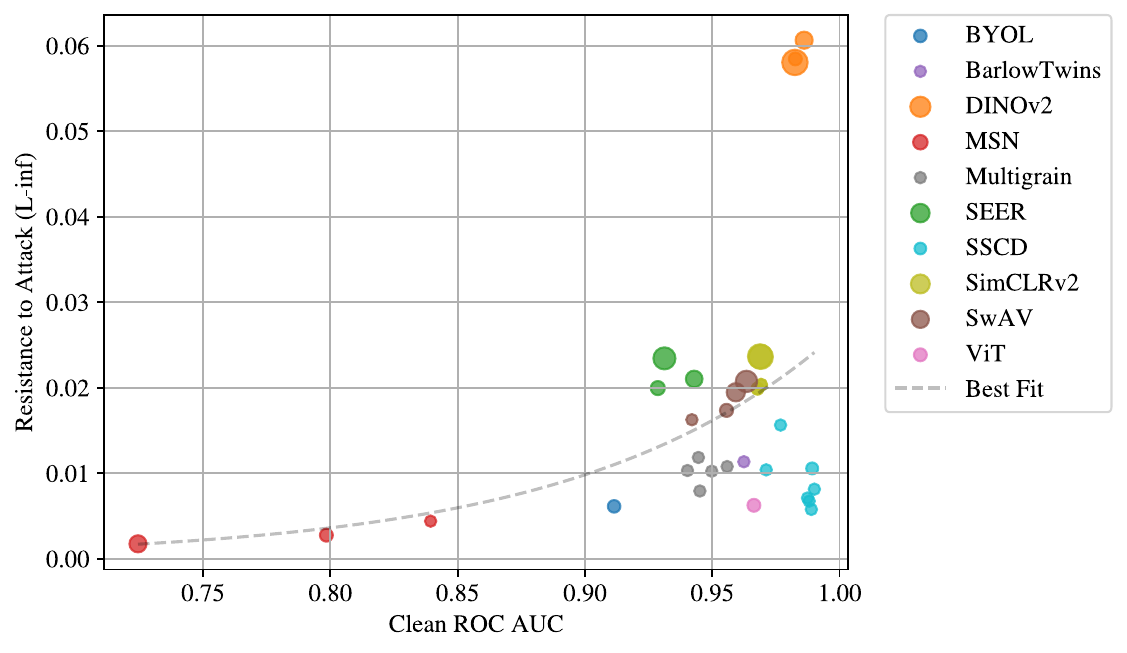}
  \label{fig:clean_roc_auc_vs_area_linf}
}
\caption{Resistance to $\ell_\infty$ attacks slightly increases with model performance. The y-axis of the right graph is calculated as the area under the corresponding curve in the left graph.}
\label{fig:attack_resistance_linf}
\end{figure*}

\textbf{Findings.} 
We find current embedding models are well-suited to similarity search in the absence of an adversary: the best performing models, DINOv2~\citep{oquab2023dinov2} and SSCD~\citep{pizzi2022self}, achieve clean ROC AUC values of 0.990 and 0.986 respectively for their best checkpoints.
Unfortunately, model performance quickly drops in the presence of a baseline adversary.
Attacking the best SSCD model with an $\epsilon = 4/255$, $\ell_\infty$ attack reduces ROC AUC to 0.057, which is completely unusable.
However we observe an interesting trend: higher performance models exhibit more resistance to adversarial attack.
The current state-of-the-art image embedding models, DINOv2, are noticeably more resistant to our fixed attack.
Note that our baseline attack gives an upper bound on robustness---targeted attacks on specific models would likely degrade performance more effectively.
Naturally, the higher resistance models result in less visibly perturbed images (see Supplement~\ref{subsec:attacked_imagery} for an example).
Even if one gains robustness to a specific threat model (e.g., $\ell_\infty$ perturbations), it may not generalize.
For example, all evaluated models are also vulnerable to $\ell_1$ perturbations (see~\Cref{fig:attack_resistance_l1}).
We expect similar results to hold for more exotic threat models.

\begin{tcolorbox}[colback=gray!10!white,leftrule=2.5mm,size=title]
\textbf{Takeaway.} 
We want a negligibly low probability of adversarial success but find that an adversary can efficiently break any evaluated embedding model.
For our scheme to be deployable, this gap would need to be closed.
\end{tcolorbox}

\subsection{High-Capacity Post-Hoc Image Watermarking}
We require a post-hoc watermark with a large capacity to encode both a cryptographic signature and a robust embedding into an image.
We find a number of ways to increase capacities or reduce capacity requirements.

\textbf{Descriptor quantization.} A ``free lunch'' optimization is to leverage existing quantization techniques. 
If the original descriptors are represented with $\mathsf{float32}$ values, using $8$-bit quantization, for example, immediately reduces capacity requirements by 75\% with little performance degradation~\citep{jacob2018quantization}.

\textbf{Conversion of generation-time watermarks to post-hoc watermarks.}
Publicly-accessible post-hoc watermarks currently do not support large enough payloads---for example, TrustMark~\citep{bui2023trustmark}, was evaluated to support at most 200 bits (at which point bit error rate (BER) becomes significant).
Despite this, the technological trend suggests that larger payloads may be supported through careful optimization of the capacity-robustness tension.

\textbf{Deferred storage.}
Instead of storing a signature and embedding directly, one can store a database reference for remote recovery of the full payload.
This optimization is necessary if a high-dimensional embedding is used (e.g., DINOv2~\citep{oquab2023dinov2} or SEER~\citep{goyal2021self} descriptors), and it introduces a critical assumption: there must exist a highly-available database.
If the database is offline, detection is impossible.
We remark that this optimization has no effect on security as it is merely an extension of the post-hoc watermark.
The adversary's capabilities are the same as in the base scheme: it can destroy but not forge watermarks.
If extra security is desired, one can encrypt the payload before storing it server-side with an encryption key that is encoded in the image.

\textbf{Stable deep hashing.}\label{para:stable_hash}
Within the image itself, we store an embedding of the image and a signature.
The embedding enables checking if a transformed image is still similar to the \textit{original} watermarked image.
This comes at a substantial cost: the embedding size is likely to be much larger than the signature, depending on the performance of the source model.
Ideally, the robust embedding function should produce identical, discretized embeddings for visually similar images.
Future research is necessary to determine whether advanced transformations can be handled by a deep stable hash, but existing schemes are already robust to quantization or re-encoding~\citep{apple2021csam}.

\begin{tcolorbox}[colback=gray!10!white,leftrule=2.5mm,size=title]
\textbf{Takeaway.} Post-hoc watermark capacity is unlikely to prevent deploying a robust, unforgeable, and publicly-detectable watermark.
\end{tcolorbox}

\section{Concluding Remarks}\label{sec:discussion}

In this paper, we explore the construction of a robust, unforgeable, and publicly-detectable watermark for images.
We prove that such a scheme exists, but its deployment is limited by the white-box security of image embedding models.
We now discuss promising avenues for future work.

\textbf{On indestructibility.}
This work focuses on unforgeability without explicitly handling \textit{indestructibility}---how easily a watermark can be removed without degrading image quality.
No scheme is presently known to have this property in the white-box setting~\citep{zhang2023watermarks}.
We remark that indestructibility may emerge organically as the set of robust transformations better approximates human vision.
Consider the ideal case where the underlying watermark and embedding function are perfectly robust.
By definition, it follows that any invalid transformation of the original watermarked image is \textit{not} visually similar and should not be detected as watermarked.
This suggests that as robustness improves, indestructibility will also improve.

\textbf{Robust embeddings in the wild.}
Robust image embeddings are applied as perceptual hashes for explicit material detection, where care must be taken when the models are deployed in a white-box setting.
As an illustrative example, Apple's NeuralHash scheme was a robust image embedding-based perceptual hash deployed to Apple systems in 2021 to detect CSAM images without revealing exact images to Apple servers~\citep{apple2021csam}.
The robust embedding model was hosted directly on user machines---as a result, adversarial attacks that broke the collision resistance of NeuralHash quickly surfaced~\citep{struppek2022learning}, which consequently broke the ``unforgeability'' of the larger system for CSAM detection.

\textbf{Adversarial robustness.}
The main barrier to deploying our robust watermark is resolving (the lack of) adversarial robustness of image embedding models.
We observe a weak correlation between raw model performance and resistance to attacks, suggesting that higher performing models may be more adversarially-robust than weaker models.
General progress in adversarial ML may be adapted to existing models to strengthen security in the white-box setting.

\textbf{Alternative pathways to public detection.}
While we have explored a method of combining deep learning and cryptography, we do not rule out the possibility of realizing a robust and publicly-detectable watermark through a different approach: we leave it to future work to develop new schemes that better resist white-box attack.
We expect, however, that any scheme leveraging deep learning may inevitably face adversarial attack: the detection algorithm must be fully public, and thus any components used within it must resist public white-box attack.

\bibliographystyle{abbrvnat}
\bibliography{references}

\begin{thebibliography}{70}
\providecommand{\natexlab}[1]{#1}
\providecommand{\url}[1]{\texttt{#1}}
\expandafter\ifx\csname urlstyle\endcsname\relax
  \providecommand{\doi}[1]{doi: #1}\else
  \providecommand{\doi}{doi: \begingroup \urlstyle{rm}\Url}\fi

\bibitem[Achiam et~al.(2023)Achiam, Adler, Agarwal, Ahmad, Akkaya, Aleman, Almeida, Altenschmidt, Altman, Anadkat, et~al.]{achiam2023gpt}
J.~Achiam, S.~Adler, S.~Agarwal, L.~Ahmad, I.~Akkaya, F.~L. Aleman, D.~Almeida, J.~Altenschmidt, S.~Altman, S.~Anadkat, et~al.
\newblock Gpt-4 technical report.
\newblock \emph{arXiv preprint arXiv:2303.08774}, 2023.

\bibitem[Alexey(2020)]{alexey2020image}
D.~Alexey.
\newblock An image is worth 16x16 words: Transformers for image recognition at scale.
\newblock \emph{arXiv preprint arXiv: 2010.11929}, 2020.

\bibitem[{Anthropic}(2024)]{anthropic2024claude}
{Anthropic}.
\newblock {The Claude 3 Model Family: Opus, Sonnet, Haiku}.
\newblock \url{https://www-cdn.anthropic.com/de8ba9b01c9ab7cbabf5c33b80b7bbc618857627/Model_Card_Claude_3.pdf}, 2024.
\newblock Online; accessed 7 August 2024.

\bibitem[{Apple}(2021)]{apple2021csam}
{Apple}.
\newblock {CSAM Detection Technical Summary}.
\newblock \url{https://www.apple.com/child-safety/pdf/CSAM_Detection_Technical_Summary.pdf}, 2021.
\newblock Online; accessed 8 August 2024.

\bibitem[Assran et~al.(2022)Assran, Caron, Misra, Bojanowski, Bordes, Vincent, Joulin, Rabbat, and Ballas]{assran2022masked}
M.~Assran, M.~Caron, I.~Misra, P.~Bojanowski, F.~Bordes, P.~Vincent, A.~Joulin, M.~Rabbat, and N.~Ballas.
\newblock Masked siamese networks for label-efficient learning.
\newblock In \emph{European Conference on Computer Vision}, pages 456--473. Springer, 2022.

\bibitem[Barreto et~al.(2003)Barreto, Lynn, and Scott]{barreto2003constructing}
P.~S. Barreto, B.~Lynn, and M.~Scott.
\newblock Constructing elliptic curves with prescribed embedding degrees.
\newblock In \emph{Security in Communication Networks: Third International Conference, SCN 2002 Amalfi, Italy, September 11--13, 2002 Revised Papers 3}, pages 257--267. Springer, 2003.

\bibitem[Berman et~al.(2019)Berman, J{\'e}gou, Vedaldi, Kokkinos, and Douze]{berman2019multigrain}
M.~Berman, H.~J{\'e}gou, A.~Vedaldi, I.~Kokkinos, and M.~Douze.
\newblock Multigrain: a unified image embedding for classes and instances.
\newblock \emph{arXiv preprint arXiv:1902.05509}, 2019.

\bibitem[Boneh et~al.(2001)Boneh, Lynn, and Shacham]{boneh2001short}
D.~Boneh, B.~Lynn, and H.~Shacham.
\newblock Short signatures from the weil pairing.
\newblock In \emph{International conference on the theory and application of cryptology and information security}, pages 514--532. Springer, 2001.

\bibitem[Bui et~al.(2023{\natexlab{a}})Bui, Agarwal, and Collomosse]{bui2023trustmark}
T.~Bui, S.~Agarwal, and J.~Collomosse.
\newblock Trustmark: Universal watermarking for arbitrary resolution images.
\newblock \emph{arXiv preprint arXiv:2311.18297}, 2023{\natexlab{a}}.

\bibitem[Bui et~al.(2023{\natexlab{b}})Bui, Agarwal, Yu, and Collomosse]{bui2023rosteals}
T.~Bui, S.~Agarwal, N.~Yu, and J.~Collomosse.
\newblock Rosteals: Robust steganography using autoencoder latent space.
\newblock In \emph{Proceedings of the IEEE/CVF Conference on Computer Vision and Pattern Recognition}, pages 933--942, 2023{\natexlab{b}}.

\bibitem[Caron et~al.(2020)Caron, Misra, Mairal, Goyal, Bojanowski, and Joulin]{caron2020unsupervised}
M.~Caron, I.~Misra, J.~Mairal, P.~Goyal, P.~Bojanowski, and A.~Joulin.
\newblock Unsupervised learning of visual features by contrasting cluster assignments.
\newblock \emph{Advances in neural information processing systems}, 33:\penalty0 9912--9924, 2020.

\bibitem[{Certicom Research}(2010)]{certicom2010sec2}
{Certicom Research}.
\newblock Sec 2: Recommended elliptic curve domain parameters.
\newblock \url{https://www.secg.org/sec2-v2.pdf}, 2010.
\newblock Online; accessed 25 September 2024.

\bibitem[Charikar(2002)]{charikar2002similarity}
M.~S. Charikar.
\newblock Similarity estimation techniques from rounding algorithms.
\newblock In \emph{Proceedings of the thiry-fourth annual ACM symposium on Theory of computing}, pages 380--388, 2002.

\bibitem[Chaudhuri et~al.(2003)Chaudhuri, Ganjam, Ganti, and Motwani]{chaudhuri2003robust}
S.~Chaudhuri, K.~Ganjam, V.~Ganti, and R.~Motwani.
\newblock Robust and efficient fuzzy match for online data cleaning.
\newblock In \emph{Proceedings of the 2003 ACM SIGMOD international conference on Management of data}, pages 313--324, 2003.

\bibitem[Chen et~al.(2020)Chen, Kornblith, Norouzi, and Hinton]{chen2020simple}
T.~Chen, S.~Kornblith, M.~Norouzi, and G.~Hinton.
\newblock A simple framework for contrastive learning of visual representations.
\newblock In \emph{International conference on machine learning}, pages 1597--1607. PMLR, 2020.

\bibitem[{Coalition for Content Provenance and Authenticity}(2023)]{c2pa2023coalition}
{Coalition for Content Provenance and Authenticity}.
\newblock {C2PA Technical Specification}.
\newblock \url{https://c2pa.org/specifications/specifications/1.3/specs/_attachments/C2PA_Specification.pdf}, 2023.
\newblock Online; accessed 16 July 2024.

\bibitem[Cox et~al.(1997)Cox, Kilian, Leighton, and Shamoon]{cox1997secure}
I.~J. Cox, J.~Kilian, F.~T. Leighton, and T.~Shamoon.
\newblock Secure spread spectrum watermarking for multimedia.
\newblock \emph{IEEE transactions on image processing}, 6\penalty0 (12):\penalty0 1673--1687, 1997.

\bibitem[{Digimarc}(1995)]{digimarc}
{Digimarc}.
\newblock Revolutionize your business with digimarc digital watermarks.
\newblock \url{https://www.digimarc.com/}, 1995.
\newblock Accessed: 2024-09-30.

\bibitem[Douze et~al.(2009)Douze, J{\'e}gou, Sandhawalia, Amsaleg, and Schmid]{douze2009evaluation}
M.~Douze, H.~J{\'e}gou, H.~Sandhawalia, L.~Amsaleg, and C.~Schmid.
\newblock Evaluation of gist descriptors for web-scale image search.
\newblock In \emph{Proceedings of the ACM international conference on image and video retrieval}, pages 1--8, 2009.

\bibitem[Dubey et~al.(2024)Dubey, Jauhri, Pandey, Kadian, Al-Dahle, Letman, Mathur, Schelten, Yang, Fan, et~al.]{dubey2024llama}
A.~Dubey, A.~Jauhri, A.~Pandey, A.~Kadian, A.~Al-Dahle, A.~Letman, A.~Mathur, A.~Schelten, A.~Yang, A.~Fan, et~al.
\newblock The llama 3 herd of models.
\newblock \emph{arXiv preprint arXiv:2407.21783}, 2024.

\bibitem[Elaine et~al.(2016)Elaine, Barker, Burr, Polk, and Smid]{elaine2016recommendation}
B.~Elaine, W.~Barker, W.~Burr, W.~Polk, and M.~Smid.
\newblock Recommendation for key management, part 1: General.
\newblock \emph{NIST Special Publication}, 800:\penalty0 57, 2016.

\bibitem[Fairoze et~al.(2023)Fairoze, Garg, Jha, Mahloujifar, Mahmoody, and Wang]{fairoze2023publicly}
J.~Fairoze, S.~Garg, S.~Jha, S.~Mahloujifar, M.~Mahmoody, and M.~Wang.
\newblock Publicly detectable watermarking for language models.
\newblock \emph{arXiv preprint arXiv:2310.18491}, 2023.

\bibitem[Fernandez et~al.(2022)Fernandez, Sablayrolles, Furon, J{\'e}gou, and Douze]{fernandez2022watermarking}
P.~Fernandez, A.~Sablayrolles, T.~Furon, H.~J{\'e}gou, and M.~Douze.
\newblock Watermarking images in self-supervised latent spaces.
\newblock In \emph{ICASSP 2022-2022 IEEE International Conference on Acoustics, Speech and Signal Processing (ICASSP)}, pages 3054--3058. IEEE, 2022.

\bibitem[Fernandez et~al.(2023)Fernandez, Couairon, J{\'e}gou, Douze, and Furon]{fernandez2023stable}
P.~Fernandez, G.~Couairon, H.~J{\'e}gou, M.~Douze, and T.~Furon.
\newblock The stable signature: Rooting watermarks in latent diffusion models.
\newblock In \emph{Proceedings of the IEEE/CVF International Conference on Computer Vision}, pages 22466--22477, 2023.

\bibitem[Ge et~al.(2013)Ge, He, Ke, and Sun]{ge2013optimized}
T.~Ge, K.~He, Q.~Ke, and J.~Sun.
\newblock Optimized product quantization for approximate nearest neighbor search.
\newblock In \emph{Proceedings of the IEEE conference on computer vision and pattern recognition}, pages 2946--2953, 2013.

\bibitem[Gegg-Harrison and Quarterman(2024)]{gegg2024ai}
W.~Gegg-Harrison and C.~Quarterman.
\newblock Ai detection's high false positive rates and the psychological and material impacts on students.
\newblock In \emph{Academic Integrity in the Age of Artificial Intelligence}, pages 199--219. IGI Global, 2024.

\bibitem[Gionis et~al.(1999)Gionis, Indyk, Motwani, et~al.]{gionis1999similarity}
A.~Gionis, P.~Indyk, R.~Motwani, et~al.
\newblock Similarity search in high dimensions via hashing.
\newblock In \emph{Vldb}, 1999.

\bibitem[{Google DeepMind}(2023)]{synthid}
{Google DeepMind}.
\newblock Identifying ai-generated content with synthid.
\newblock \url{https://deepmind.google/technologies/synthid/}, 2023.
\newblock Accessed: 2024-09-30.

\bibitem[Goyal et~al.(2021)Goyal, Caron, Lefaudeux, Xu, Wang, Pai, Singh, Liptchinsky, Misra, Joulin, et~al.]{goyal2021self}
P.~Goyal, M.~Caron, B.~Lefaudeux, M.~Xu, P.~Wang, V.~Pai, M.~Singh, V.~Liptchinsky, I.~Misra, A.~Joulin, et~al.
\newblock Self-supervised pretraining of visual features in the wild.
\newblock \emph{arXiv preprint arXiv:2103.01988}, 2021.

\bibitem[{GPTZero}(2023)]{gptzero}
{GPTZero}.
\newblock {GPTZero}.
\newblock \url{hhttps://gptzero.me/}, 2023.
\newblock Online; accessed 7 August 2024.

\bibitem[Grill et~al.(2020)Grill, Strub, Altch{\'e}, Tallec, Richemond, Buchatskaya, Doersch, Avila~Pires, Guo, Gheshlaghi~Azar, et~al.]{grill2020bootstrap}
J.-B. Grill, F.~Strub, F.~Altch{\'e}, C.~Tallec, P.~Richemond, E.~Buchatskaya, C.~Doersch, B.~Avila~Pires, Z.~Guo, M.~Gheshlaghi~Azar, et~al.
\newblock Bootstrap your own latent-a new approach to self-supervised learning.
\newblock \emph{Advances in neural information processing systems}, 33:\penalty0 21271--21284, 2020.

\bibitem[Ha et~al.(2024)Ha, Passananti, Bhaskar, Shan, Southen, Zheng, and Zhao]{ha2024organic}
A.~Y.~J. Ha, J.~Passananti, R.~Bhaskar, S.~Shan, R.~Southen, H.~Zheng, and B.~Y. Zhao.
\newblock Organic or diffused: Can we distinguish human art from ai-generated images?
\newblock \emph{arXiv preprint arXiv:2402.03214}, 2024.

\bibitem[Haitsma and Kalker(2002)]{haitsma2002highly}
J.~Haitsma and T.~Kalker.
\newblock A highly robust audio fingerprinting system.
\newblock In \emph{Ismir}, volume 2002, pages 107--115, 2002.

\bibitem[Indyk and Motwani(1998)]{indyk1998approximate}
P.~Indyk and R.~Motwani.
\newblock Approximate nearest neighbors: towards removing the curse of dimensionality.
\newblock In \emph{Proceedings of the thirtieth annual ACM symposium on Theory of computing}, pages 604--613, 1998.

\bibitem[Jacob et~al.(2018)Jacob, Kligys, Chen, Zhu, Tang, Howard, Adam, and Kalenichenko]{jacob2018quantization}
B.~Jacob, S.~Kligys, B.~Chen, M.~Zhu, M.~Tang, A.~Howard, H.~Adam, and D.~Kalenichenko.
\newblock Quantization and training of neural networks for efficient integer-arithmetic-only inference.
\newblock In \emph{Proceedings of the IEEE conference on computer vision and pattern recognition}, pages 2704--2713, 2018.

\bibitem[Johnson et~al.(2001)Johnson, Menezes, and Vanstone]{johnson2001elliptic}
D.~Johnson, A.~Menezes, and S.~Vanstone.
\newblock The elliptic curve digital signature algorithm (ecdsa).
\newblock \emph{International journal of information security}, 1:\penalty0 36--63, 2001.

\bibitem[Katz and Lindell(2014)]{katzlindell}
J.~Katz and Y.~Lindell.
\newblock \emph{Introduction to Modern Cryptography, Second Edition}.
\newblock Chapman \& Hall/CRC, 2nd edition, 2014.
\newblock ISBN 1466570261.

\bibitem[Laurie(2014)]{laurie2014certificate}
B.~Laurie.
\newblock Certificate transparency.
\newblock \emph{Communications of the ACM}, 57\penalty0 (10):\penalty0 40--46, 2014.

\bibitem[Madry et~al.(2018)Madry, Makelov, Schmidt, Tsipras, and Vladu]{madry2018towards}
A.~Madry, A.~Makelov, L.~Schmidt, D.~Tsipras, and A.~Vladu.
\newblock Towards deep learning models resistant to adversarial attacks.
\newblock In \emph{International Conference on Learning Representations}, 2018.

\bibitem[{Midjourney}(2022)]{midjourney}
{Midjourney}.
\newblock Midjourney.
\newblock \url{https://www.midjourney.com/home}, 2022.

\bibitem[Muyco and Hernandez(2019)]{muyco2019least}
S.~D. Muyco and A.~A. Hernandez.
\newblock Least significant bit hash algorithm for digital image watermarking authentication.
\newblock In \emph{Proceedings of the 2019 5th International Conference on Computing and Artificial Intelligence}. Association for Computing Machinery, 2019.
\newblock \doi{10.1145/3330482.3330523}.
\newblock URL \url{https://doi.org/10.1145/3330482.3330523}.

\bibitem[Navas et~al.(2008)Navas, Ajay, Lekshmi, Archana, and Sasikumar]{navas2008dwt}
K.~Navas, M.~C. Ajay, M.~Lekshmi, T.~S. Archana, and M.~Sasikumar.
\newblock Dwt-dct-svd based watermarking.
\newblock In \emph{2008 3rd international conference on communication systems software and middleware and workshops (COMSWARE'08)}, pages 271--274. IEEE, 2008.

\bibitem[Oord et~al.(2018)Oord, Li, and Vinyals]{oord2018representation}
A.~v.~d. Oord, Y.~Li, and O.~Vinyals.
\newblock Representation learning with contrastive predictive coding.
\newblock \emph{arXiv preprint arXiv:1807.03748}, 2018.

\bibitem[Oquab et~al.(2023)Oquab, Darcet, Moutakanni, Vo, Szafraniec, Khalidov, Fernandez, Haziza, Massa, El-Nouby, et~al.]{oquab2023dinov2}
M.~Oquab, T.~Darcet, T.~Moutakanni, H.~Vo, M.~Szafraniec, V.~Khalidov, P.~Fernandez, D.~Haziza, F.~Massa, A.~El-Nouby, et~al.
\newblock Dinov2: Learning robust visual features without supervision.
\newblock \emph{arXiv preprint arXiv:2304.07193}, 2023.

\bibitem[{Pan, Titusz}(2022)]{iscc2024enhancement}
{Pan, Titusz}.
\newblock {ISCC - Enhancement Proposals (IEPs)}.
\newblock \url{https://ieps.iscc.codes/}, 2022.
\newblock Online; accessed 16 July 2024.

\bibitem[Paszke et~al.(2019)Paszke, Gross, Massa, Lerer, Bradbury, Chanan, Killeen, Lin, Gimelshein, Antiga, et~al.]{paszke2019pytorch}
A.~Paszke, S.~Gross, F.~Massa, A.~Lerer, J.~Bradbury, G.~Chanan, T.~Killeen, Z.~Lin, N.~Gimelshein, L.~Antiga, et~al.
\newblock Pytorch: An imperative style, high-performance deep learning library.
\newblock \emph{Advances in neural information processing systems}, 32, 2019.

\bibitem[Pizzi et~al.(2022)Pizzi, Roy, Ravindra, Goyal, and Douze]{pizzi2022self}
E.~Pizzi, S.~D. Roy, S.~N. Ravindra, P.~Goyal, and M.~Douze.
\newblock A self-supervised descriptor for image copy detection.
\newblock In \emph{Proceedings of the IEEE/CVF Conference on Computer Vision and Pattern Recognition}, pages 14532--14542, 2022.

\bibitem[Prokos et~al.(2023)Prokos, Fendley, Green, Schuster, Tromer, Jois, and Cao]{prokos2023squint}
J.~Prokos, N.~Fendley, M.~Green, R.~Schuster, E.~Tromer, T.~Jois, and Y.~Cao.
\newblock Squint hard enough: Attacking perceptual hashing with adversarial machine learning.
\newblock In \emph{32nd USENIX Security Symposium (USENIX Security 23)}, pages 211--228, 2023.

\bibitem[Reid et~al.(2024)Reid, Savinov, Teplyashin, Lepikhin, Lillicrap, Alayrac, Soricut, Lazaridou, Firat, Schrittwieser, et~al.]{reid2024gemini}
M.~Reid, N.~Savinov, D.~Teplyashin, D.~Lepikhin, T.~Lillicrap, J.-b. Alayrac, R.~Soricut, A.~Lazaridou, O.~Firat, J.~Schrittwieser, et~al.
\newblock Gemini 1.5: Unlocking multimodal understanding across millions of tokens of context.
\newblock \emph{arXiv preprint arXiv:2403.05530}, 2024.

\bibitem[Rivest et~al.(1978)Rivest, Shamir, and Adleman]{rivest1978method}
R.~Rivest, A.~Shamir, and L.~Adleman.
\newblock A method for obtaining digital signatures and public-key cryptosystems.
\newblock \emph{Communications of the ACM}, 21\penalty0 (2):\penalty0 120--126, 1978.

\bibitem[Rombach et~al.(2022)Rombach, Blattmann, Lorenz, Esser, and Ommer]{rombach2022high}
R.~Rombach, A.~Blattmann, D.~Lorenz, P.~Esser, and B.~Ommer.
\newblock High-resolution image synthesis with latent diffusion models.
\newblock In \emph{Proceedings of the IEEE/CVF conference on computer vision and pattern recognition}, pages 10684--10695, 2022.

\bibitem[Saltman and Thorley(2021)]{saltman2021practical}
E.~Saltman and T.~Thorley.
\newblock Practical and technical considerations.
\newblock \emph{Broadening the GIFCT Hash-Sharing Database Taxonomy: An Assessment and Recommended Next Steps}, page~12, 2021.

\bibitem[Seo et~al.(2004)Seo, Haitsma, Kalker, and Yoo]{seo2004robust}
J.~S. Seo, J.~Haitsma, T.~Kalker, and C.~D. Yoo.
\newblock A robust image fingerprinting system using the radon transform.
\newblock \emph{Signal Processing: Image Communication}, 19\penalty0 (4):\penalty0 325--339, 2004.

\bibitem[{Steg}(2019)]{stegai}
{Steg}.
\newblock Forensic watermarking for digital media.
\newblock \url{https://steg.ai/}, 2019.
\newblock Accessed: 2024-09-30.

\bibitem[{StopNCII}(2024)]{stopncii}
{StopNCII}.
\newblock Stop non-consensual intimate image abuse.
\newblock \url{https://stopncii.org/}, 2024.
\newblock Accessed: 2024-09-30.

\bibitem[Struppek et~al.(2022)Struppek, Hintersdorf, Neider, and Kersting]{struppek2022learning}
L.~Struppek, D.~Hintersdorf, D.~Neider, and K.~Kersting.
\newblock Learning to break deep perceptual hashing: The use case neuralhash.
\newblock In \emph{Proceedings of the 2022 ACM Conference on Fairness, Accountability, and Transparency}, pages 58--69, 2022.

\bibitem[Tancik et~al.(2020)Tancik, Mildenhall, and Ng]{tancik2020stegastamp}
M.~Tancik, B.~Mildenhall, and R.~Ng.
\newblock Stegastamp: Invisible hyperlinks in physical photographs.
\newblock In \emph{Proceedings of the IEEE/CVF conference on computer vision and pattern recognition}, pages 2117--2126, 2020.

\bibitem[{Turnitin}(1998)]{turnitin}
{Turnitin}.
\newblock Plagiarism prevention trusted by educators worldwide.
\newblock \url{https://www.turnitin.com/products/similarity/}, 1998.
\newblock Accessed: 2024-09-30.

\bibitem[Van Den~Oord et~al.(2017)Van Den~Oord, Vinyals, et~al.]{van2017neural}
A.~Van Den~Oord, O.~Vinyals, et~al.
\newblock Neural discrete representation learning.
\newblock \emph{Advances in neural information processing systems}, 30, 2017.

\bibitem[Venkatesan et~al.(2000)Venkatesan, Koon, Jakubowski, and Moulin]{venkatesan2000robust}
R.~Venkatesan, S.-M. Koon, M.~H. Jakubowski, and P.~Moulin.
\newblock Robust image hashing.
\newblock In \emph{Proceedings 2000 International Conference on Image Processing (Cat. No. 00CH37101)}, volume~3, pages 664--666. IEEE, 2000.

\bibitem[Wan et~al.(2022)Wan, Wang, Zhang, Li, Yu, and Sun]{wan2022comprehensive}
W.~Wan, J.~Wang, Y.~Zhang, J.~Li, H.~Yu, and J.~Sun.
\newblock A comprehensive survey on robust image watermarking.
\newblock \emph{Neurocomputing}, 488:\penalty0 226--247, 2022.

\bibitem[Wang(2006)]{wang2006shazam}
A.~Wang.
\newblock The shazam music recognition service.
\newblock \emph{Communications of the ACM}, 49\penalty0 (8):\penalty0 44--48, 2006.

\bibitem[Wang et~al.(2003)]{wang2003industrial}
A.~Wang et~al.
\newblock An industrial strength audio search algorithm.
\newblock In \emph{Ismir}, volume 2003, pages 7--13. Washington, DC, 2003.

\bibitem[Wen et~al.(2023)Wen, Kirchenbauer, Geiping, and Goldstein]{wen2023tree}
Y.~Wen, J.~Kirchenbauer, J.~Geiping, and T.~Goldstein.
\newblock Tree-ring watermarks: Fingerprints for diffusion images that are invisible and robust.
\newblock \emph{arXiv preprint arXiv:2305.20030}, 2023.

\bibitem[Wolfgang and Delp(1996)]{wolfgang1996watermark}
R.~B. Wolfgang and E.~J. Delp.
\newblock A watermark for digital images.
\newblock In \emph{Proceedings of 3rd IEEE International Conference on Image Processing}, volume~3, pages 219--222. IEEE, 1996.

\bibitem[Yang et~al.(2024)Yang, Zeng, Chen, Fang, Zhang, and Yu]{yang2024gaussian}
Z.~Yang, K.~Zeng, K.~Chen, H.~Fang, W.~Zhang, and N.~Yu.
\newblock Gaussian shading: Provable performance-lossless image watermarking for diffusion models.
\newblock In \emph{Proceedings of the IEEE/CVF Conference on Computer Vision and Pattern Recognition}, pages 12162--12171, 2024.

\bibitem[Yeo and Kim(2003)]{yeo2003generalized}
I.-K. Yeo and H.~J. Kim.
\newblock Generalized patchwork algorithm for image watermarking.
\newblock \emph{Multimedia systems}, 9:\penalty0 261--265, 2003.

\bibitem[Zbontar et~al.(2021)Zbontar, Jing, Misra, LeCun, and Deny]{zbontar2021barlow}
J.~Zbontar, L.~Jing, I.~Misra, Y.~LeCun, and S.~Deny.
\newblock Barlow twins: Self-supervised learning via redundancy reduction.
\newblock In \emph{International conference on machine learning}, pages 12310--12320. PMLR, 2021.

\bibitem[Zhang et~al.(2023)Zhang, Edelman, Francati, Venturi, Ateniese, and Barak]{zhang2023watermarks}
H.~Zhang, B.~L. Edelman, D.~Francati, D.~Venturi, G.~Ateniese, and B.~Barak.
\newblock Watermarks in the sand: Impossibility of strong watermarking for generative models.
\newblock \emph{arXiv preprint arXiv:2311.04378}, 2023.

\bibitem[Zhu et~al.(2018)Zhu, Kaplan, Johnson, and Fei-Fei]{zhu2018hidden}
J.~Zhu, R.~Kaplan, J.~Johnson, and L.~Fei-Fei.
\newblock Hidden: Hiding data with deep networks.
\newblock In \emph{Proceedings of the European conference on computer vision (ECCV)}, pages 657--672, 2018.

\end{thebibliography}

\clearpage
\appendix

\section{Extended Related Work}\label{app:extended_rw}

Here we expand on the watermarking and robust embedding literature.

\subsection{Watermarking}

The literature on watermarking can be grouped into two categories: classical watermarking and deep watermarking. 
We highlight key papers as follows.

\textbf{Classical watermarking.} Classical watermarking primarily aims to embed watermarks to enforce copyright.
An early idea was to embed secret information into the least significant bits of each pixel, leading to a low-distortion watermark~\citep{wolfgang1996watermark}---this is similar to our simple non-robust scheme.
Like ours, since the watermark is embedded in the `least important'' bits, the watermark is non-robust.
\citet{cox1997secure} developed the spread-spectrum technique where the watermark is instead embedded in the frequency domain.
This led to robustness to re-encoding, minor crops, and lossy compression.
Many subsequent works followed with the aim of increasing payload size and robustness.
The patchwork algorithm~\citep{yeo2003generalized} and DWT-DCT-SVD~\citep{navas2008dwt} are particularly well adopted solutions.

\textbf{Deep watermarking.} 
More recent watermarks turned to deep learning for better performance.
Such watermarks demonstrated higher payloads, higher robustness, and lower image distortion than classical methods~\citep{wan2022comprehensive}.
HiDDeN~\citep{zhu2018hidden} was the first encoder-decoder watermark for images.
StegaStamp~\citep{tancik2020stegastamp} refined the encoder-decoder architecture by leveraging spatial information. 
This led to better robustness against geometric transformations.
SSL~\citep{fernandez2022watermarking} uses a different approach and instead watermarks images in the latent space during inference---this trades off image quality for inference speed.
RoSteALS~\citep{bui2023rosteals} similarly applies the watermark in the latent space but instead uses a VQ-VAE architecture~\citep{van2017neural}.
Finally, TrustMark~\citep{bui2023trustmark} is resolution invariant and supports re-watermarking out of the box.

More recently, watermarking schemes have been tailor-made for specific generation processes.
In the image setting, the Stable Signature~\citep{fernandez2023stable} introduced an active strategy to embed the watermark during diffusion. 
The latent decoder is fine-tuned to embed a binary signature which a pre-trained watermark extractor can identify. 
This lower-level approach led to higher performance than post-hoc methods.
The Tree-Ring method~\citep{wen2023tree} also performs watermarking in the latent space and is designed to minimize distortions at the cost of model performance.
Gaussian Shading~\citep{yang2024gaussian} is the latest image watermark for diffusion models.
The approach is claimed to be computationally distortion-free: the watermarked distribution is computationally close to the native output distribution.
In the text setting,~\citet{fairoze2023publicly} gave an unforgeable and publicly-detectable watermark with weak robustness.

\subsection{Robust Embeddings}\label{subsec:robust_embeddings}

There are a wide range of problems that come under the umbrella of robust embeddings.
The fundamental task in any one of these formulations is to determine if two objects (images) are similar or not from a succinct representation of the original objects.
If this core problem is solvable, then it can be applied to various applications such as perceptual hashing, fingerprinting, or copy detection.
We summarize the literature for each of these tracks.

\textbf{Perceptual hashing.}
\citet{indyk1998approximate} introduced locality-sensitive hashing (LSH).
\citet{gionis1999similarity} gave a more efficient construction that requires asymptotically fewer queries when applied to nearest neighbor search.
Random hyperplane-based perceptual hashing is the most widespread method of performing LSHs~\citep{charikar2002similarity}.
The technique involves dividing up the embedding space by randomly partitioning it with hyperplanes---each cell produced by the partitioning process is assigned a unique hash value.
This method was adopted by Apple for their NeuralHash system~\citep{apple2021csam}.
To hash an image using NeuralHash, features are first extracted using a convolutional neural network.
The features are then perceptually hashed with a hyperplane-based LSH.
\citet{struppek2022learning} later demonstrated that it is easy to perturb a arbitrary image to be closely embedded to any arbitrary different image, breaking NeuralHash completely.
ITQ~\citep{ge2013optimized} is a deep learning approach to achieving the same effect as a classical LSH, i.e. ITQ produces a binary hash.

\textbf{Fingerprinting.}
\citet{venkatesan2000robust} introduced the notion of a robust image fingerprint and used wavelet representations of images to gain robustness to common transformations.
\citet{seo2004robust} improved their approach by instead using Radon transforms for robustness to affine trasformations.
In the audio setting, similar results exist~\citep{haitsma2002highly} that make optimizations specific to the audio domain.
In a similar vein, \citet{wang2003industrial} developed a highly effective audio fingerprint for identifying music.
Their techniques formed the backbone of Shazam~\citep{wang2006shazam}.

\textbf{Feature extraction.}
We highlight key works that are relevant to this paper. 
We refer the reader to the DINOv2 paper~\citep{oquab2023dinov2} for a more comprehensive overview of the technical state of image feature extraction.
The Simple Framework for Contrastive Learning of Visual Representations (SimCLR)~\citep{chen2020simple} demonstrated that many prior self-supervised learning algorithms could be simplified, removing the need for specialized architectures or memory banks.
This simplification led to better performing models that required fewer weights than prior approaches.
\citet{pizzi2022self} augmented SimCLR with entropy regularization, InfoNCE loss~\citep{oord2018representation}, and inference-time score normalization to develop SSCD.
The model was recently used by the DINOv2 and Llama 3 teams to identify duplicate images~\citep{oquab2023dinov2,dubey2024llama}.
\citet{berman2019multigrain} developed Multigrain: a network architecture designed to produce compact descriptors for image classification and object retrieval downstream tasks.
It crucially leverages different levels of image ``granularity'' to learn generalized features.
Bootstrap Your Own Latents (BYOL)~\citep{grill2020bootstrap} uses two sub-networks, denoted the online and target networks.
The online network learns to predict the target network representation of the same image under a particular transformation.
BYOL descriptors were found to outperform SimCLR descriptors.
SwAV~\citep{caron2020unsupervised} is designed to learn constrastively without directly making pairwise comparisons.
The method aims to cluster transformations of the same image together: this has a similar effect as direct constrastive comparison.
The approach was also found to outperform SimCLR.
\citet{zbontar2021barlow} developed Barlow Twins. 
The method aims to avoid learning trivial solutions (e.g., constant embeddings) by optimizing cross-correlations between outputs of two identical networks (each fed with transformed images) to be close to a target cross-correlation.
Aside from avoiding collapse, this has an additional effect of minimizing redundant embedding information.
Barlow Twins descriptors were found to outperform BYOL and SwAV descriptors.
SEER~\citep{goyal2021self} is another self-supervised learning method that further closed the gap with supervised methods.
Their models are based on SwAV and are optimized to scale---SEER achieved new state-of-the-art across a range of model size classes, scaling to billions of parameters.
\citet{assran2022masked} developed Masked Siamese Networks (MSN).
The framework aims to match the representation of masked images with the original image.
This naturally interfaces with Vision Transformers (ViT)~\citep{alexey2020image} as they only handle unmasked images by default.
This combination was found to scale well and outperformed prior work.
DINOv2~\citep{oquab2023dinov2} is another ViT-based self-supervised method that aimed to consolidate research following SEER to optimize performance at scale.
They demonstrate that careful orchestration of the training pipeline coupled with recent advancements in self-supervision led to new state-of-the-art embeddings.
In particular,DINOv2 image patches were shown to semantically capture various aspects of human vision.

\section{Formalism}\label{app:formalism}

We present the necessary primitives for formal analysis of robust and publicly-detectable watermarking schemes ($\RPWS$).
In doing so, we precisely define what properties each primitive must satisfy---this allows us to parameterize the resulting $\RPWS$ later on.

\subsection{Definitions}

\begin{definition}[Robust embedding function]\label{def:ref}

A $(\mathcal{T}, m, n, \epsilon)$-robust embedding function $(\REF)$ is a triple of algorithms $(\Gen, \Embed)$ defined as follows:

\begin{enumerate}
    \item $\Gen(1^\lambda, m, n, \mathcal{T}) \to (\Embed, \Compare)$. Generate takes the input size $m$, output size $n$, and the security parameter $\lambda$, outputting the robust embedding function in time $\poly(\lambda)$.
    
    \item $\Embed(x) \to e$. The function $\Embed : \bits^{m} \to \bits^{n}$ takes an arbitrary object $x$ and maps it to a discrete point in the embedding space $e$.

    \item $\Compare(x, y) \to b$. The function $\Compare: \bits^n \times \bits^n \to \bits$ takes two embeddings $x$ and $y$ and outputs a single bit $b$ such that $b$ is $\true$ when $x$ and $y$ are similar, and $\false$ otherwise.
\end{enumerate}

A valid robust embedding function must satisfy the following two properties:

\begin{enumerate}
    \item \textbf{Correctness.} For all $x$ and transformations $T \in \mathcal{T}$,
    \begin{align*}
        \Pr\left[
        \begin{array}{ccc}
            \begin{matrix}
                \Compare(\Embed(x), \Embed(T(x))) = \true: \\
                \Embed, \Compare \gets \Gen(1^\lambda, m, n, \mathcal{T}) \\
            \end{matrix} \\
        \end{array}
        \right] \geq 1 - \epsilon.
    \end{align*}
    
    \item \textbf{Collision-resistance.} For any choice of $x$, it must be that
    \begin{align*}
        \Pr\left[
        \begin{array}{ccc}
            \begin{matrix}
                \Compare(\Embed(x), \Embed(x^*)) = \true \\ 
                \land\ x^* \not\in \Gamma(x): \\
                \Embed, \Compare \gets \Gen(1^\lambda, m, n, \mathcal{T}) \\
                x^* \gets \mathcal{A}^{\Embed(\cdot)}(m, n, \mathcal{T}, x) \\
            \end{matrix} \\
        \end{array}
        \right] \leq \epsilon.
    \end{align*}
    
\end{enumerate}

\end{definition}

\begin{definition}[Cryptographic signature scheme]\label{def:crypto_sig}

A $(\delta, \lambda)$-cryptographic signature scheme $\CSS$ is a 3-tuple ($\Gen, \Sign, \Verify)$ defined as follows:
\begin{enumerate}
    \item $\Gen(1^{\lambda}) \to (sk, pk)$. Takes in the target security bits number $\lambda$ and outputs a fresh secret key $sk$ and public key $pk$ pair in time $\poly(\lambda)$.
    \item $\Sign(sk, x) \to \sigma$. Given the secret key $sk$ and a object to sign $x$, $\Sign$ computes $\sigma \in \bits^\delta$: a signature of $sk$ on $x$.
    \item $\Verify(pk, x, \sigma) \to b$. Given a public key $pk$, object $x$, and signature $\sigma$, $\Verify$ checks if $\sigma$ is a valid signature of the corresponding secret key of $pk$ on $x$ and encodes the result into one bit $b$. It returns either $\true$ or $\false$ depending on if verification succeeded or not.
\end{enumerate}

A valid $\CSS$ needs to satisfy the following properties:

\begin{enumerate}
    \item \textbf{Correctness.} It holds that for all $x$,
    \begin{align*}
        \Pr\left[
        \begin{array}{ccc}
            \begin{matrix}
                \Verify(pk, x, \Sign(sk, x)) = \true: \\
                (sk, pk) \gets \Gen(1^\lambda) \\
            \end{matrix} \\
        \end{array}
        \right] \geq 1 - \negl(\lambda).
    \end{align*}
    
    \item \textbf{Unforgeability}. Let $X$ denote the set of oracle queries that the adversary makes to $\Sign(sk, \cdot)$. It holds that for all $x$,
    \begin{align*}
        \Pr\left[
        \begin{array}{ccc}
            \begin{matrix}
                \Verify(pk, x^*, \sigma^*) = \true \\
                \land\ x^* \not\in X: \\
                (sk, pk) \gets \Gen(1^\lambda) \\
                (x^*, \sigma^*) \gets \mathcal{A}^{\Sign(sk, \cdot)}(pk) \\
            \end{matrix} \\
        \end{array}
        \right] \leq \negl(\lambda).
    \end{align*}
\end{enumerate}

\end{definition}

\begin{definition}[Post-hoc watermarking scheme]

A $(\mathcal{T}, c, \epsilon)$-post-hoc watermarking scheme $\PGWS$ is a 3-tuple of algorithms defined as follows:

\begin{enumerate}
    \item $\Gen(1^\lambda, c, \mathcal{T}) \to (\Encode, \Decode)$. Takes in the target security bits number $\lambda$, the capacity size in bits $c$, and the set of possible transformations $\mathcal{T}$ and outputs the encode and decode algorithms in time $\poly(\lambda)$.
    
    \item $\Encode(x, m) \to x'$. Given a message $m$ such that $m \in \bits^c$ and target object $x$, $\Encode$ plants $m$ a watermark in $x$ producing $x'$.
    
    \item $\Decode(x^*) \to m$. Given a potentially watermarked object $x^*$, $\Decode$ recovers $m \in \bits^c$ if it exists.
\end{enumerate}

Let $\mathcal{X}$ and $\mathcal{M}$ be the set of all valid objects and messages respectively.
A valid $\PGWS$ must satisfy correctness with optional private unforgeability:

\begin{enumerate}
    \item \textbf{Correctness.} For all choices of $x$,  $m$, and $\mathcal{T}$,
    \begin{align*}
        \Pr\left[
        \begin{array}{ccc}
            \begin{matrix}
                \Decode(T(\Encode(x, m))) = m: \\
                \Encode, \Decode \gets \Gen(1^\lambda, c, \mathcal{T}) \\
            \end{matrix} \\
        \end{array}
        \right] \geq 1 - \epsilon.
    \end{align*}
    
    \item \textbf{Private Unforegability.} Let $X$ be the list of object queries that $\mathcal{A}$ makes to the decoding oracle and let $\Gamma(X)$ be the set containing all neighboring query objects. It must be that for all $x$, $m$, and $\mathcal{A}$,
    \begin{align*}
        \Pr\left[
        \begin{array}{ccc}
            \begin{matrix}
                \Decode(x^*, m) = m \\ 
                \land\ x^* \not\in \Gamma(X): \\
                \textunderscore, \Decode \gets \Gen(1^\lambda, c, \mathcal{T}) \\
                x^* \gets \mathcal{A}^{\Decode(\cdot, \cdot)}(c, \mathcal{T}, x, m) \\
            \end{matrix} \\
        \end{array}
        \right] \leq \epsilon.
    \end{align*}
    Note that this property is not a strict requirement for the purpose of constructing a robust and publicly-detectable watermark.
    We include a definition of private unforgeability for completeness.
\end{enumerate}
\end{definition}

\begin{definition}[Robust publicly-detectable watermarking scheme]

A $(\mathcal{T}, \epsilon)$-robust publicly-detectable watermarking scheme $\RPWS$ consists of three algorithms defined as follows:

\begin{enumerate}
    \item $\Gen(1^\lambda) \to (sk, pk)$. Takes in the security parameter $\lambda$ and outputs a fresh secret key $sk$ and public key $pk$ pair.
    
    \item $\Watermark(sk, x) \to x'$. Given a secret key $sk$ and target object $x$, $\Watermark$ plants a watermark in $x$ under $sk$ producing $x'$.
    
    \item $\Detect(pk, x^*) \to b$. Given a potentially watermarked object $x^*$, $\Detect$ checks if there is valid watermark under the corresponding secret key to $pk$. It outputs $\true$ or $\false$ as a single bit depending on if verification succeeds or not.
\end{enumerate}

A valid $\RPWS$ must satisfy the following properties:

\begin{enumerate}
    \item \textbf{Correctness.}
    It holds that for all $x$,
    \begin{align*}
        \Pr\left[
        \begin{array}{ccc}
            \begin{matrix}
                \Detect(pk, T(\Watermark(sk, x))) = \true: \\
                (sk, pk) \gets \Gen(1^\lambda) \\
            \end{matrix} \\
        \end{array}
        \right] \geq 1 - \epsilon.
    \end{align*}
    
    \item \textbf{Unforegability.}
    Let $X$ denote the set of oracle queries that the adversary makes to $\Watermark(sk, \cdot)$.
    It holds that for all $x$,
    \begin{align*}
        \Pr\left[
        \begin{array}{ccc}
            \begin{matrix}
                \Detect(pk, x^*) = \true \\
                \land\ x^* \not\in \Gamma(X): \\
                (sk, pk) \gets \Gen(1^\lambda) \\
                (x^*) \gets \mathcal{A}^{\Watermark(sk, \cdot)}(pk) \\
            \end{matrix} \\
        \end{array}
        \right] \leq \epsilon.
    \end{align*}
\end{enumerate}

\end{definition}

\subsection{Constructing a Robust and Publicly-Detectable Watermark}\label{app:embed_to_watermark}

In this section we will prove the following theorem:

\main*

\begin{proof}
To construct a $\PGWS$, we compose the $\REF$, $\CSS$ and $\PGWS$ in the natural way: the $\REF$ is used to compute a stable representation (embedding) of the image.
This embedding is cryptographically signed using $\CSS$ and both the signature and embedding are encoded within the image using the $\PGWS$.
In order to detect the watermark in an arbitrary image, it suffices to (a) check that the embedding of the image at hand is similar to the encoded embedding, and (b) check that the signature is authentic.
We formalize this simple construction as follows.

$\Gen$. 
On input $1^\lambda$, run $\CSS.\Gen(1^\lambda)$ to obtain $sk$ and $pk$. 
Output $(sk, pk)$.

$\Watermark$. 
Given the secret key $sk$ and cover image $x \in \mathcal{X}$,
\begin{enumerate}
    \item Compute the robust embedding $e = \REF.\Embed(x)$.
    \item Sign the embedding $\sigma \gets \CSS.\Sign(sk, e)$.
    \item Plant $\sigma$ and $e$ into $x$: $x' \gets \PGWS.\Encode(x, \sigma \parallel e)$.
    \item Output $x'$.
\end{enumerate}

$\Detect$. Given a public key $pk$ and candidate watermarked object $x'$,
\begin{enumerate}
    \item Compute the aggregated hash $e = \REF.\Embed(x')$.
    \item Decode the embedded payload if it exists $\sigma' \parallel e' \gets \PGWS.\Decode(x')$.
    \item Attempt signature verification $b_1 \gets \CSS.\Verify(pk, e', \sigma')$.
    \item Attempt closeness verification $b_2 \gets \REF.\Compare(e, e')$.
    \item Output $b_1 \land b_2$.
\end{enumerate}

\begin{claim}\label{clm:wm_corr}
The above scheme is correct if the underlying signature scheme and post-generation watermarking scheme are both correct.
\end{claim}

\begin{proof}
Observe that this $\RPWS$ is robust to the intersection of the input transformation sets, i.e., $\mathcal{T}_\RPWS = \mathcal{T}_\REF \cap \mathcal{T}_\PGWS$---any transformation applied to the image that is not in this set will result in either incorrect decoding using the $\PGWS$ or $\REF.\Compare$ outputting $\false$.
Correctness immediately follows from the construction: we calculate the success probability loss as follows.
Given $(sk, pk) \gets \Gen(1^\lambda)$, $\Detect(pk, x) \Longleftrightarrow \CSS.\Verify(pk, \cdot, \cdot) \land \REF.\Compare(\cdot, \cdot)$ must output true for any honestly generated $x$.
Signature verification fails with probability $\negl(\lambda)$.
The similarity check fails with probability $\epsilon_\REF$.
Decoding fails with probability $\epsilon_\PGWS$.
Thus the overall check will succeed with probability at least $1 - (\epsilon_\REF + \epsilon_\PGWS + \negl(\lambda)) =: 1 - \epsilon_\PGWS$.
\end{proof}

\begin{claim}\label{clm:wm_forge_to_sig_or_ref}
If the scheme is forgeable then either the underlying signature scheme is forgeable or the robust embedding function is not collision-resistant.
\end{claim}

\begin{proof}
Assume for the sake of contradiction that there exists an adversary $\mathcal{A}_{\RPWS}$ that can break the unforgeability of the above $\RPWS$.
We will show that there exists an adversary that can produce an object $x^*$ such that $\RPWS.\Detect(x^*) = \true$ that was not honestly-generated or within the neighborhood of transformed honestly-generated objects.
Consider the following two reductions---we will prove that at least one of them will succeed in their respective security game if $\mathcal{A}_\RPWS$ is successful in breaking the $\RPWS$.

First, we construct an adversary $\mathcal{A}_\CSS$ against the underlying signature scheme's unforgeability.

Construction of $\mathcal{A}_\CSS$. On input $pk$ (and oracle access to the private signing algorithm):
\begin{enumerate}
    \item Run $\mathcal{A}_{\RPWS}(pk)$ to obtain $x^*$ such that with high probability $\RPWS.\Detect(pk, x^*) = \true$.
    \item Compute $\sigma^* \parallel e^* \gets \PGWS.\Decode(x^*)$.
    \item Return $(e^*, \sigma^*)$ as a forgery.
\end{enumerate}

Similarly, we construct $\mathcal{A}_\REF$ that breaks the collision-resistance of the underlying robust embedding function as follows:

Construction of $\mathcal{A}_\REF$. On input $x$, $\mathcal{T}$:
\begin{enumerate}
    \item Run $\mathcal{A}_{\RPWS}(pk)$ to obtain $x^*$ such that with high probability $\Detect(pk, x^*) = \true$.
    \item Compute the output of $\REF.\Embed$: $e \gets \REF.
    \Embed(x^*)$.
    \item Decode the payload $\sigma^* \parallel e^* \gets \Decode(x^*)$
    \item Submit $(e, e^*)$ as a collision.
\end{enumerate}

We know that with high probability, $\Detect(pk, x^*) = \CSS.\Verify(pk, e^*, \sigma^*) \ \land \ \REF.\Compare(e, e^*) = \true$ implying that either $\Verify$ or $\Compare$ has been forged such that $x^* \not\in \Gamma(x)$.
It is easy to see that at least one of $\mathcal{A}_\CSS$ or $\mathcal{A}_\REF$ will succeed: in order for $\Detect$ to be $\true$ for some $x^*$ that is not in the neighborhood of an honestly-watermarked image, the adversary must have forged a signature or broken the embedding function's collision resistance (or both).
Since we know that $\REF$ and $\CSS$ are respectively secure, we have our contradiction and it holds that the $\RPWS$ is unforgeable.
\end{proof}

Finally, by assumption, the capacity $c$ of the watermark is large enough to support a signature and a robust embedding, i.e., $c \geq n + \delta = \abs{\sigma} + \abs{e}$ and so it is valid to store both objects within the image using the $\PGWS$.

Combining~\Cref{clm:wm_corr} and~\Cref{clm:wm_forge_to_sig_or_ref} yields the proof.
\end{proof}

\section{Experimental Data}\label{sec:experimental_data}

We provide additional data from our evaluation of state-of-the-art pre-trained image embedding models to instantiate the robust embedding function.

\textbf{Platform.} All model checkpoints were evaluated using PyTorch~\citep{paszke2019pytorch} on a virtual machine running Debian 11 equipped with a single NVIDIA A100 GPU with 40GB VRAM.

\textbf{Models.} We enumerate all evaluated models in~\Cref{tab:models}.

\begin{table}[ht]
\centering
\resizebox{1.0\textwidth}{!}{%
\begin{tabular}{@{}l|llll@{}}
\toprule
\textit{Model}       & \textbf{Architecture} & \textbf{Data} & \textbf{Parameter Count} & \textbf{Dimensions}                                          \\ \midrule
\textit{BYOL}~\citep{grill2020bootstrap}        & ResNet-50             & ImageNet      & 72M                      & 2048                                                         \\
\textit{DINOv2}~\citep{oquab2023dinov2}      & ViT-g/14              & LVD           & 86M -- 1136M             & 257 $\times$ 768 -- 257 $\times$ 1536                        \\
\textit{SEER}~\citep{goyal2021self}        & RG256                 & IG            & 141M -- 637M             & 3712 $\times$ 12 $\times$ 12 -- 7392 $\times$ 12 $\times$ 12 \\
\textit{MSN}~\citep{assran2022masked}         & ViT-L/7               & ImageNet      & 85M -- 303M              & 197 $\times$ 384 -- 197 $\times$ 1024                        \\
\textit{BarlowTwins}~\citep{zbontar2021barlow} & ResNet-50             & ImageNet      & 25M                      & 1000                                                         \\
\textit{SwAV}~\citep{caron2020unsupervised}        & RX101-32x16d          & ImageNet      & 25M -- 586M              & 1000 -- 10240                                                \\
\textit{ViT}~\citep{alexey2020image}         & Vit-B/16              & JFT           & 86M                      & 197 $\times$ 768                                             \\
\textit{Multigrain}~\citep{berman2019multigrain}  & ResNet-50             & ImageNet      & 25M                      & 2048                                                         \\
\textit{SimCLRv2}~\citep{chen2020simple}    & ResNet-152x3          & ImageNet      & 25M -- 800M              & 2048 -- 6144                                                 \\
\textit{SSCD}~\citep{pizzi2022self}        & ResNeXt101-32x4       & DISC          & 24M -- 44M               & 512 -- 1024                                                 
\end{tabular}%
}
\caption{Image embedding model metadata. If multiple architectures or checkpoints are evaluated for a given model, the minimum and maximum parameter counts and dimensions are reported along with the largest architecture.}
\label{tab:models}
\end{table}

\begin{figure}[ht]
\centering
\subfigure{
  \includegraphics[height=0.31\textwidth]{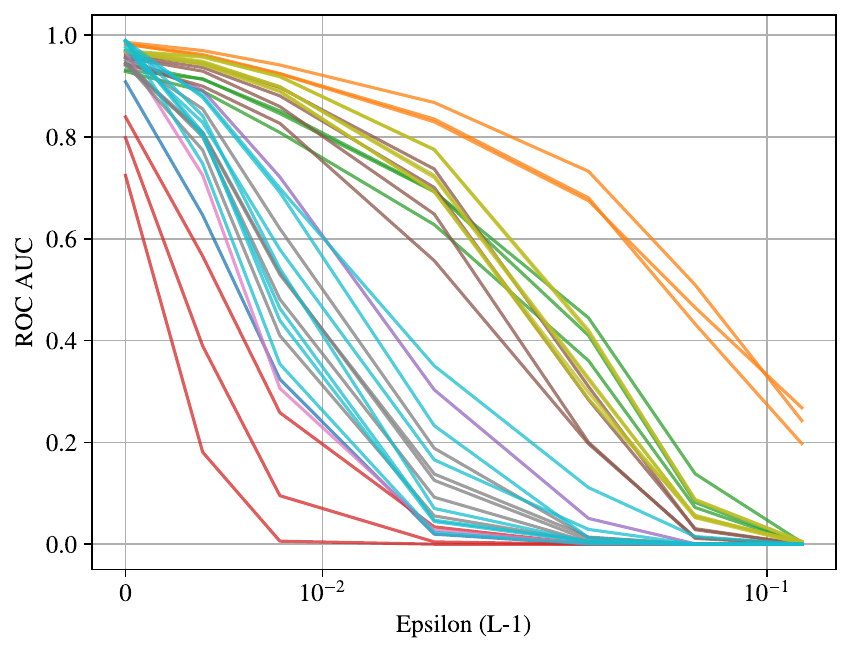}
  \label{fig:roc_auc_vs_epsilon_l1}
}
\subfigure{
  \includegraphics[height=0.31\textwidth]{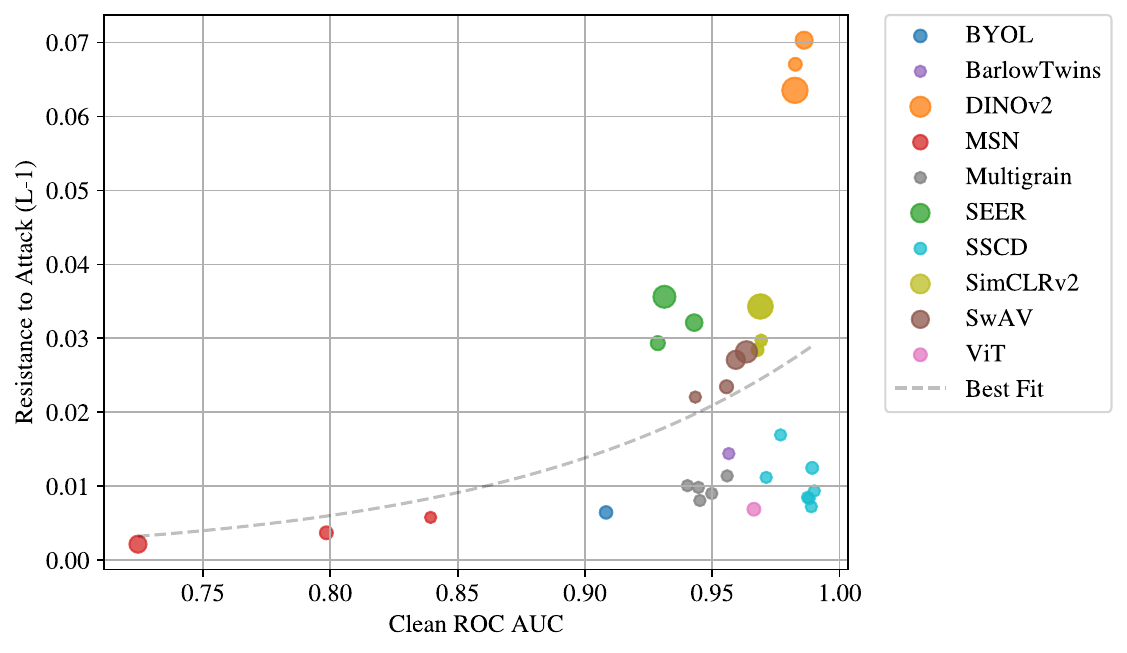}
  \label{fig:clean_roc_auc_vs_area_l1}
}
\caption{Resistance to $\ell_1$ attacks slightly increases with model performance. The y-axis of the right figure is calculated as the area under the corresponding curve in the left figure.}
\label{fig:attack_resistance_l1}
\end{figure}

\clearpage

\subsection{Attacked Imagery}\label{subsec:attacked_imagery}

Recall that our PGD-based adversarial attack seeks to perform the following for a given image triple ($I_a$, $I_{a'}$, $I_b$): for the positive pair $I_a$ and $I_{a'}$ that are visually similar, it aims to force $\Embed(I_a)$ and $\Embed(I_{a'})$ to have a \textit{low} $\ell_2$-normalized dot product.
For the negative pair $I_a$ and $I_b$ that are not visually similar, it aims to force $\Embed(I_a)$ and $\Embed(I_b)$ to have a \textit{high} $\ell_2$-normalized dot product.
We provide examples of attack perturbations against DINOv2 (giant) and SSCD (imagenet-advanced) checkpoints for a select image triple below.
In~\Cref{fig:base} we show the base image $I_a$.
In~Figures 7, 9, 11 and 13, we show examples of attacked $I_{a'}$ images for various epsilon choices.
In~Figures 8, 10, 12 and 14, we show examples of attacked $I_{b}$ images for various epsilon choices.
Note that $I_{a'}$ is a ``cropped and affine-transformed'' version of $I_a$ and $I_b$ is a completely unrelated image.

\begin{figure}[ht]
    \centering
    \includegraphics[scale=0.1]{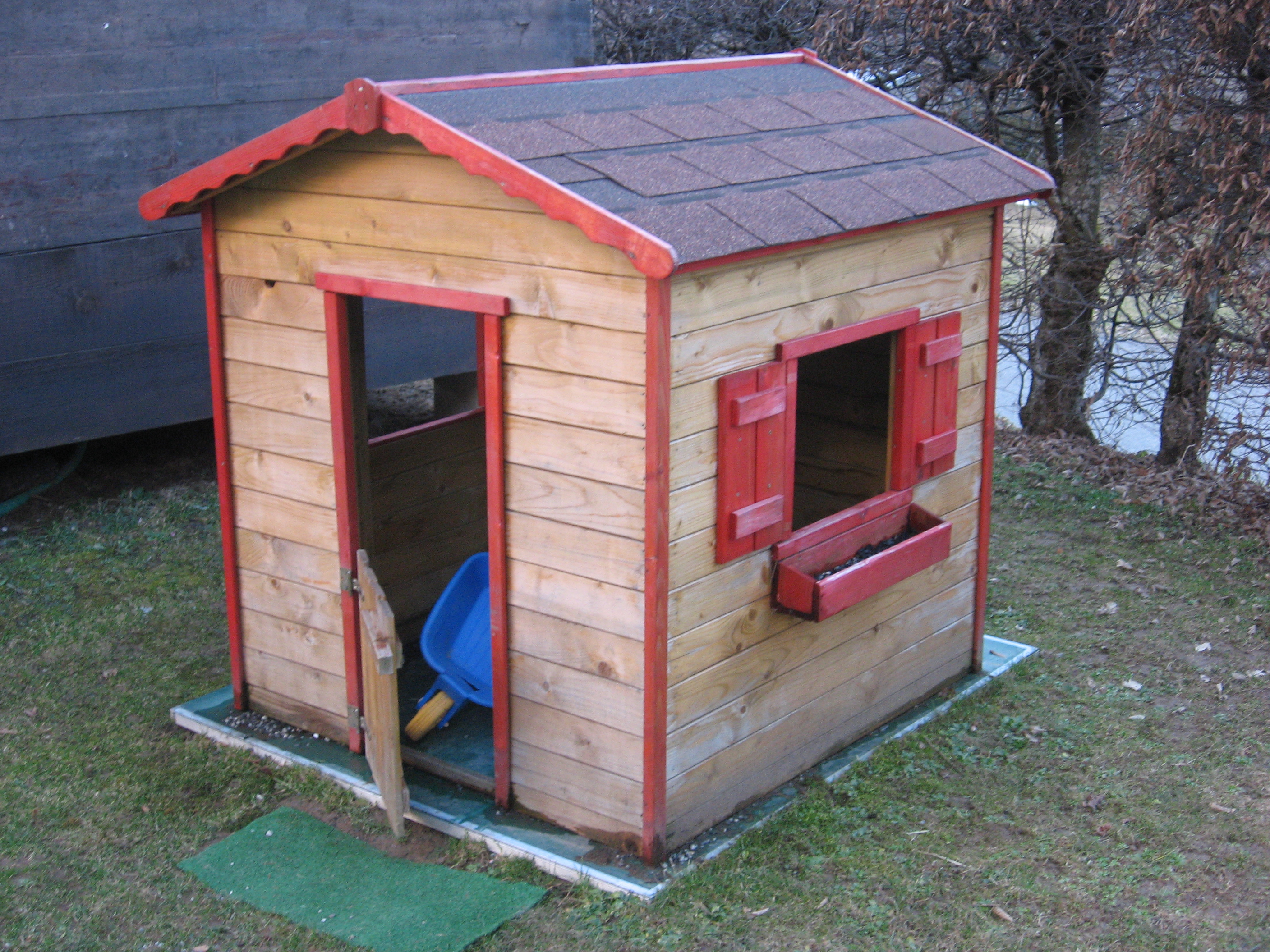}
    \caption{Base image.}
    \label{fig:base}
\end{figure}

\begin{figure}[ht]
  \centering
  \begin{minipage}{0.125\textwidth}
    \includegraphics[scale=0.3]{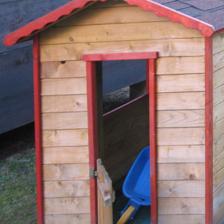}
  \end{minipage}
  \hfill
  \begin{minipage}{0.125\textwidth}
    \includegraphics[scale=0.3]{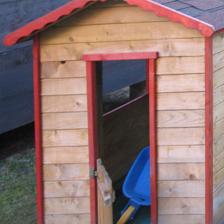}
  \end{minipage}
  \hfill
  \begin{minipage}{0.125\textwidth}
    \includegraphics[scale=0.3]{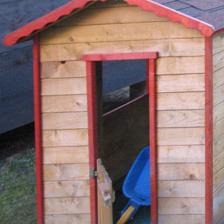}
  \end{minipage}
  \hfill
  \begin{minipage}{0.125\textwidth}
    \includegraphics[scale=0.3]{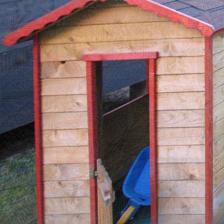}
  \end{minipage}
  \hfill
  \begin{minipage}{0.125\textwidth}
    \includegraphics[scale=0.3]{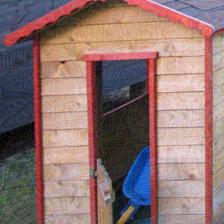}
  \end{minipage}
  \hfill
  \begin{minipage}{0.125\textwidth}
    \includegraphics[scale=0.3]{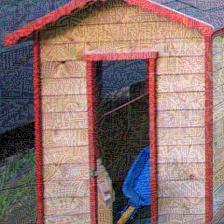}
  \end{minipage}
  \hfill
  \begin{minipage}{0.125\textwidth}
    \includegraphics[scale=0.3]{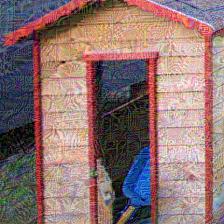}
  \end{minipage}
  \label{fig:dino_linf_positive_examples}
  \caption{$\ell_\infty$-attacked positive image for $\epsilon \in \{0/255, 1/255, 2/255, 4/255, 8/255, 16/255, 32/255\}$ over DINOv2 (giant).}
  \vspace{10pt}
  \begin{minipage}{0.125\textwidth}
    \includegraphics[scale=0.3]{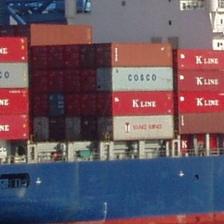}
  \end{minipage}
  \hfill
  \begin{minipage}{0.125\textwidth}
    \includegraphics[scale=0.3]{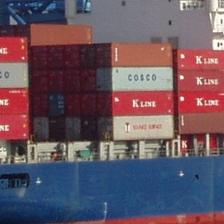}
  \end{minipage}
  \hfill
  \begin{minipage}{0.125\textwidth}
    \includegraphics[scale=0.3]{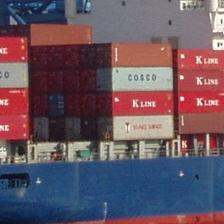}
  \end{minipage}
  \hfill
  \begin{minipage}{0.125\textwidth}
    \includegraphics[scale=0.3]{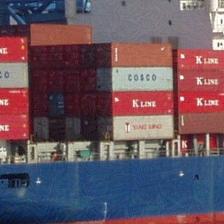}
  \end{minipage}
  \hfill
  \begin{minipage}{0.125\textwidth}
    \includegraphics[scale=0.3]{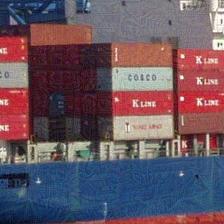}
  \end{minipage}
  \hfill
  \begin{minipage}{0.125\textwidth}
    \includegraphics[scale=0.3]{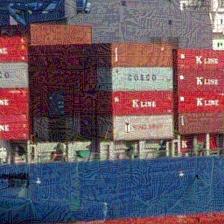}
  \end{minipage}
  \hfill
  \begin{minipage}{0.125\textwidth}
    \includegraphics[scale=0.3]{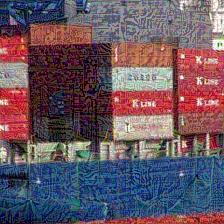}
  \end{minipage}
  \label{fig:dino_linf_negative_examples}
  \caption{$\ell_\infty$-attacked negative image for $\epsilon \in \{0/255, 1/255, 2/255, 4/255, 8/255, 16/255, 32/255\}$ over DINOv2 (giant).}
  \vspace{10pt}
  \begin{minipage}{0.125\textwidth}
    \includegraphics[scale=0.3]{Images/dinov2-giant_eps=0_t.jpg}
  \end{minipage}
  \hfill
  \begin{minipage}{0.125\textwidth}
    \includegraphics[scale=0.3]{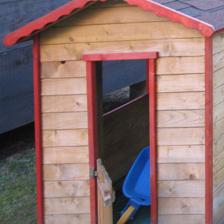}
  \end{minipage}
  \hfill
  \begin{minipage}{0.125\textwidth}
    \includegraphics[scale=0.3]{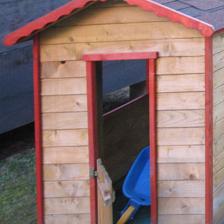}
  \end{minipage}
  \hfill
  \begin{minipage}{0.125\textwidth}
    \includegraphics[scale=0.3]{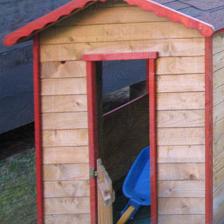}
  \end{minipage}
  \hfill
  \begin{minipage}{0.125\textwidth}
    \includegraphics[scale=0.3]{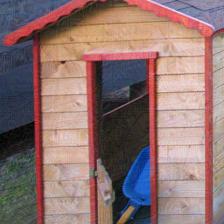}
  \end{minipage}
  \hfill
  \begin{minipage}{0.125\textwidth}
    \includegraphics[scale=0.3]{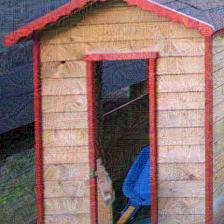}
  \end{minipage}
  \hfill
  \begin{minipage}{0.125\textwidth}
    \includegraphics[scale=0.3]{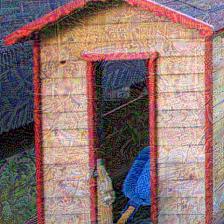}
  \end{minipage}
  \label{fig:dino_l1_positive_examples}
  \caption{$\ell_1$-attacked positive image for $\epsilon \in \{0/255, 1/255, 2/255, 4/255, 8/255, 16/255, 32/255\}$ over DINOv2 (giant).}
  \vspace{10pt}
  \begin{minipage}{0.125\textwidth}
    \includegraphics[scale=0.3]{Images/dinov2-giant_eps=0_f.jpg}
  \end{minipage}
  \hfill
  \begin{minipage}{0.125\textwidth}
    \includegraphics[scale=0.3]{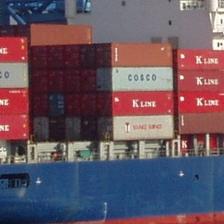}
  \end{minipage}
  \hfill
  \begin{minipage}{0.125\textwidth}
    \includegraphics[scale=0.3]{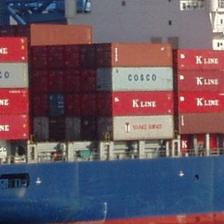}
  \end{minipage}
  \hfill
  \begin{minipage}{0.125\textwidth}
    \includegraphics[scale=0.3]{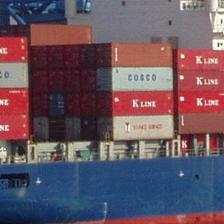}
  \end{minipage}
  \hfill
  \begin{minipage}{0.125\textwidth}
    \includegraphics[scale=0.3]{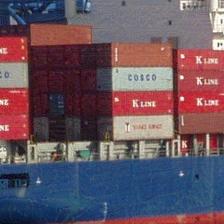}
  \end{minipage}
  \hfill
  \begin{minipage}{0.125\textwidth}
    \includegraphics[scale=0.3]{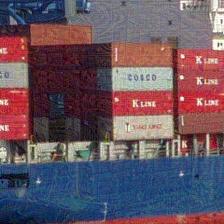}
  \end{minipage}
  \hfill
  \begin{minipage}{0.125\textwidth}
    \includegraphics[scale=0.3]{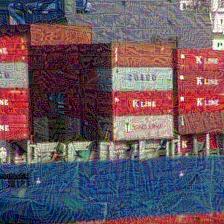}
  \end{minipage}
  \label{fig:dino_l1_negative_examples}
  \caption{$\ell_1$-attacked negative image for $\epsilon \in \{0/255, 1/255, 2/255, 4/255, 8/255, 16/255, 32/255\}$ over DINOv2 (giant).}
\end{figure}

\begin{figure}[ht]
  \centering
  \begin{minipage}{0.125\textwidth}
    \includegraphics[scale=0.235]{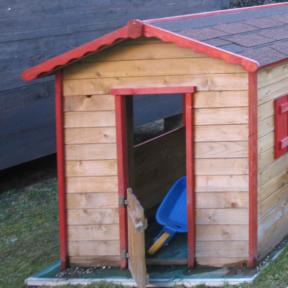}
  \end{minipage}
  \hfill
  \begin{minipage}{0.125\textwidth}
    \includegraphics[scale=0.235]{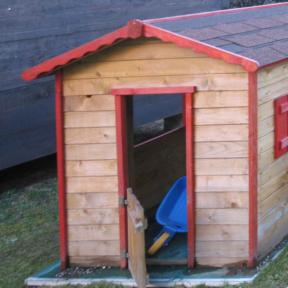}
  \end{minipage}
  \hfill
  \begin{minipage}{0.125\textwidth}
    \includegraphics[scale=0.235]{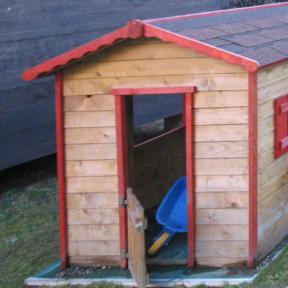}
  \end{minipage}
  \hfill
  \begin{minipage}{0.125\textwidth}
    \includegraphics[scale=0.235]{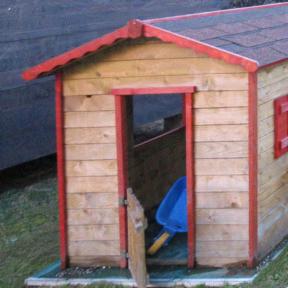}
  \end{minipage}
  \hfill
  \begin{minipage}{0.125\textwidth}
    \includegraphics[scale=0.235]{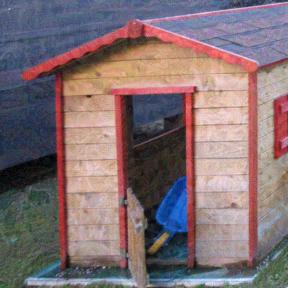}
  \end{minipage}
  \hfill
  \begin{minipage}{0.125\textwidth}
    \includegraphics[scale=0.235]{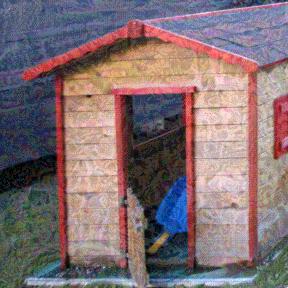}
  \end{minipage}
  \hfill
  \begin{minipage}{0.125\textwidth}
    \includegraphics[scale=0.235]{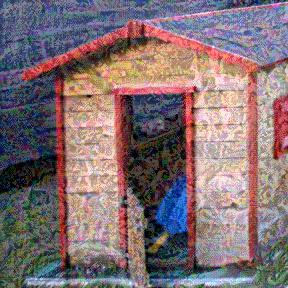}
  \end{minipage}
  \label{fig:sscd_linf_positive_examples}
  \caption{$\ell_\infty$-attacked positive image for $\epsilon \in \{0/255, 1/255, 2/255, 4/255, 8/255, 16/255, 32/255\}$ over SSCD (imagenet-advanced).}
  \vspace{10pt}
  \begin{minipage}{0.125\textwidth}
    \includegraphics[scale=0.235]{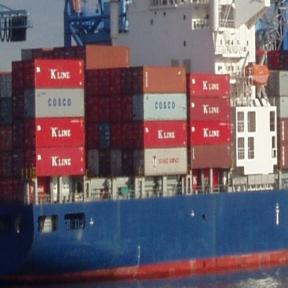}
  \end{minipage}
  \hfill
  \begin{minipage}{0.125\textwidth}
    \includegraphics[scale=0.235]{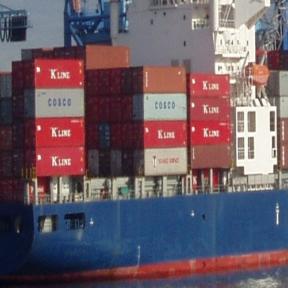}
  \end{minipage}
  \hfill
  \begin{minipage}{0.125\textwidth}
    \includegraphics[scale=0.235]{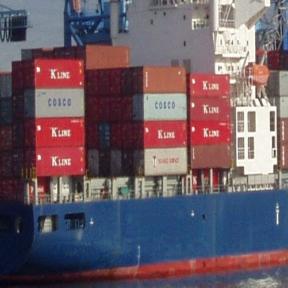}
  \end{minipage}
  \hfill
  \begin{minipage}{0.125\textwidth}
    \includegraphics[scale=0.235]{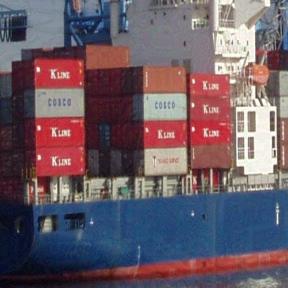}
  \end{minipage}
  \hfill
  \begin{minipage}{0.125\textwidth}
    \includegraphics[scale=0.235]{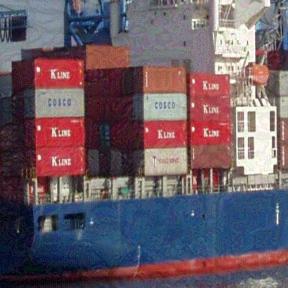}
  \end{minipage}
  \hfill
  \begin{minipage}{0.125\textwidth}
    \includegraphics[scale=0.235]{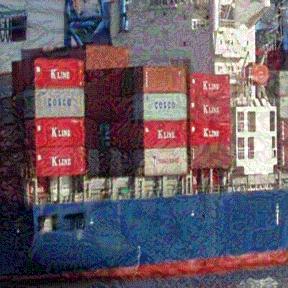}
  \end{minipage}
  \hfill
  \begin{minipage}{0.125\textwidth}
    \includegraphics[scale=0.235]{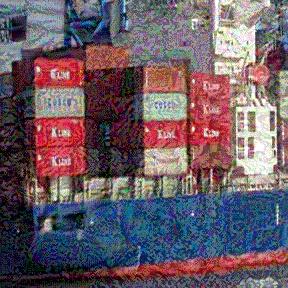}
  \end{minipage}
  \label{fig:sscd_linf_negative_examples}
  \caption{$\ell_\infty$-attacked negative image for $\epsilon \in \{0/255, 1/255, 2/255, 4/255, 8/255, 16/255, 32/255\}$ over SSCD (imagenet-advanced).}
  \vspace{10pt}
  \begin{minipage}{0.125\textwidth}
    \includegraphics[scale=0.235]{Images/sscd-ia_eps=0_t.jpg}
  \end{minipage}
  \hfill
  \begin{minipage}{0.125\textwidth}
    \includegraphics[scale=0.235]{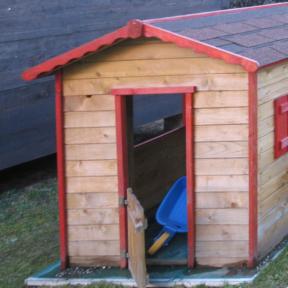}
  \end{minipage}
  \hfill
  \begin{minipage}{0.125\textwidth}
    \includegraphics[scale=0.235]{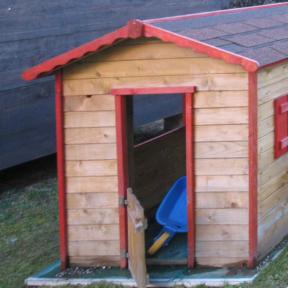}
  \end{minipage}
  \hfill
  \begin{minipage}{0.125\textwidth}
    \includegraphics[scale=0.235]{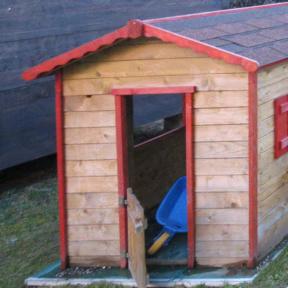}
  \end{minipage}
  \hfill
  \begin{minipage}{0.125\textwidth}
    \includegraphics[scale=0.235]{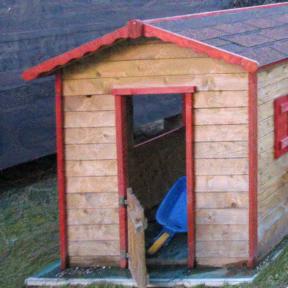}
  \end{minipage}
  \hfill
  \begin{minipage}{0.125\textwidth}
    \includegraphics[scale=0.235]{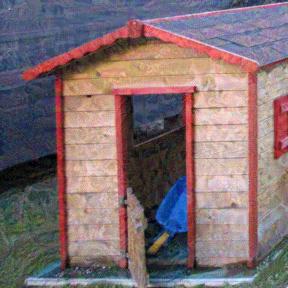}
  \end{minipage}
  \hfill
  \begin{minipage}{0.125\textwidth}
    \includegraphics[scale=0.235]{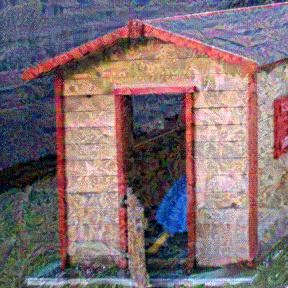}
  \end{minipage}
  \label{fig:sscd_l1_positive_examples}
  \caption{$\ell_1$-attacked positive image for $\epsilon \in \{0/255, 1/255, 2/255, 4/255, 8/255, 16/255, 32/255\}$ over SSCD (imagenet-advanced).}
  \vspace{10pt}
  \begin{minipage}{0.125\textwidth}
    \includegraphics[scale=0.235]{Images/sscd-ia_eps=0_f.jpg}
  \end{minipage}
  \hfill
  \begin{minipage}{0.125\textwidth}
    \includegraphics[scale=0.235]{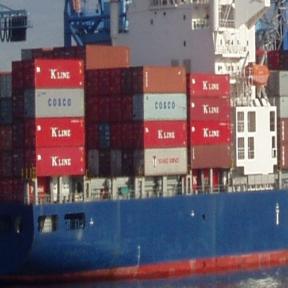}
  \end{minipage}
  \hfill
  \begin{minipage}{0.125\textwidth}
    \includegraphics[scale=0.235]{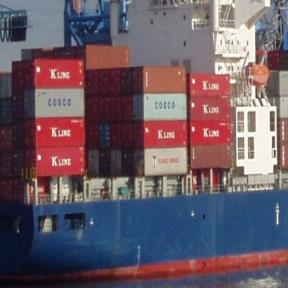}
  \end{minipage}
  \hfill
  \begin{minipage}{0.125\textwidth}
    \includegraphics[scale=0.235]{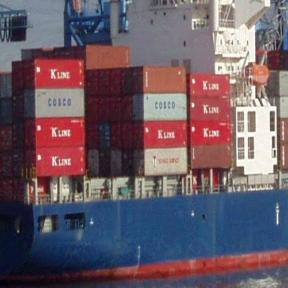}
  \end{minipage}
  \hfill
  \begin{minipage}{0.125\textwidth}
    \includegraphics[scale=0.235]{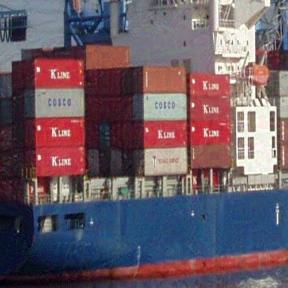}
  \end{minipage}
  \hfill
  \begin{minipage}{0.125\textwidth}
    \includegraphics[scale=0.235]{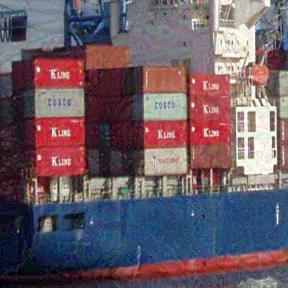}
  \end{minipage}
  \hfill
  \begin{minipage}{0.125\textwidth}
    \includegraphics[scale=0.235]{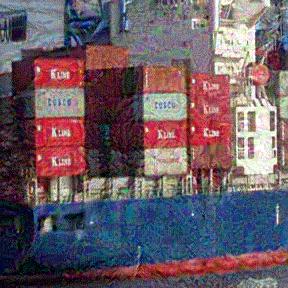}
  \end{minipage}
  \label{fig:sscd_l1_negative_examples}
  \caption{$\ell_1$-attacked negative image for $\epsilon \in \{0/255, 1/255, 2/255, 4/255, 8/255, 16/255, 32/255\}$ over SSCD (imagenet-advanced).}
\end{figure}

\end{document}